\documentclass{birkjour}
\usepackage[square,sort,comma,numbers]{natbib}

\usepackage{latexsym}
\usepackage{enumerate}
\usepackage[font=scriptsize,labelfont=bf]{caption}
\usepackage{pst-all}
\usepackage{multido}
\usepackage{pst-blur}
\usepackage{color}
\usepackage{epstopdf}

\usepackage{graphicx}
\usepackage{amsmath}
\usepackage{amssymb}

\usepackage{amsfonts,amsgen,amsopn}

\usepackage{url}

\newtheorem{theorem}{Theorem}[section]
\newtheorem{lemma}[theorem]{Lemma}

\newtheorem{corollary}[theorem]{Corollary}

\newtheorem{observation}[theorem]{Observation}
\newtheorem{claim}[theorem]{Claim}

\newcommand{\T}{{\mathcal T}}
\newcommand{\C}{{\mathcal C}}
\newcommand{\NN}{{\mathbb N}}

 \theoremstyle{definition}
 
 \theoremstyle{remark}

 \numberwithin{equation}{section}

\begin{document}

%
%
%
%
%
%
%
%
%

\title[Computing MP distance between binary phylogenetic trees]{On the complexity of computing MP\\ distance between binary phylogenetic trees}


\author{Steven Kelk}
\address{Department of Knowledge Engineering (DKE)\\ Maastricht University\\ P.O. Box 616, 6200
MD Maastricht\\ The Netherlands}
\email{steven.kelk@maastrichtuniversity.nl}


\author{Mareike Fischer}

\address{Ernst-Moritz-Arndt University of Greifswald \\ Department for Mathematics and Computer Science \\ Walther-Rathenau-Str. 47 \\ 17487 Greifswald, Germany}
\email{email@mareikefischer.de}




\date{\today}


\begin{abstract}
Within the field of phylogenetics there is great interest in distance measures to quantify the dissimilarity of two trees. Recently, a new distance measure has been proposed: the Maximum Parsimony (MP) distance. This is based on the difference of the parsimony scores of a single character on both trees under consideration, and the goal is to find the character which maximizes this difference. Here we show that computation of MP distance on two \emph{binary} phylogenetic trees is NP-hard. This is a highly nontrivial extension of an earlier NP-hardness proof for two multifurcating phylogenetic trees, and it is particularly relevant given the prominence of binary trees in the phylogenetics literature. As a corollary to the main hardness result we show that computation of MP distance is also hard on binary trees if the number of states available is bounded. In fact, via a different reduction we show that it is hard even if only two states are available. Finally, as a first response to this hardness we give a simple Integer Linear Program (ILP) formulation which is capable of computing the MP distance exactly for small trees (and for larger trees when only a small number of character states are available) and which is used to computationally verify several auxiliary results required by the hardness proofs.

\end{abstract}
\maketitle

\pagebreak
\section{Introduction}\label{sec:intro}

When present day species are considered and their evolutionary relationships are to be investigated, phylogeneticists often seek to estimate the best evolutionary tree explaining the given species data (e.g. DNA alignments). However, it is well known that different data sets on the same species can lead to different trees, or that different phylogenetic tree estimation methods, like e.g. Maximum Parsimony or Maximum Likelihood or distance based methods, can lead to different trees even for the same data set \citep{husonsteel,fischerthatte2009}. Thus, in practice one is often confronted with multiple trees, and it is therefore interesting to measure how different these trees really are. A new way of determining their relative similarity has recently been proposed \citep{fischer2014maximum}: the Maximum Parsimony distance (or MP distance, for short).  

This metric basically requires the search for a character which has a low parsimony score on one of the trees involved and a high score on the other one. In \citep{fischer2014maximum} it has been shown that calculating the MP distance between two trees is NP-hard. The proof presented there required non-binary trees (sometimes also called \emph{multifurcating} trees). This was not entirely satisfactory, for the following reason. In many branches of phylogenetics multifurcating trees are used to model uncertainty about the precise order of branching events \citep{maddison1989reconstructing}, in which case the term \emph{unresolved} is often used instead of multifurcating. Distance measures which interpret multifurcations this way often have the property that the distance decreases, or in the worst case stays the same, if one or both of the input trees become more unresolved \citep{lv2013practical}. However, the parsimony score of a single tree increases (or in the best case stays the same) if its edges are contracted to create multifurcations. This is why algorithms
that compute Maximum Parsimony trees usually output binary trees: a non-binary solution can be refined into a binary solution without loss of quality. Given this traditional emphasis on binary trees in the parsimony literature, and the fact that evolutionary events such as mutation or speciation are understood to split a lineage into two parts, not three or more \citep{haws2013}, it is logical to explore the complexity of MP distance on binary trees.

In this paper, we answer this question by showing that computing the MP distance between two binary trees is, unfortunately, also NP-hard. This is by no means a simple extension of the hardness proofs in \citep{fischer2014maximum}. To prove hardness in the present case we are required to develop a rather elaborate array of novel gadgets and arguments, with a strong graph-theoretical flavour. 

Moreover, we show as a corollary to the main theorem that this hardness remains if we restrict the number of character states to four (or more). Note that this covers the most important biological applications, as the DNA and RNA alphabets consist of four character states each, and the protein alphabet consists of 20 states. However, when morphological data is analyzed, binary characters are also often relevant, which is why we consider this case, too. We show that when restricted to two character states, calculating the MP distance is not only also NP-hard, but even APX-hard, which means that there exists a constant $c>1$ such that a polynomial-time $c$-approximation is impossible unless $P = NP$.

As a tentative first step towards addressing the NP-hardness of the MP distance, we present a simple Integer Linear Program (ILP) which calculates this distance (both on a bounded number of states as well as in the unbounded case). The ILP  is rather ``explicit'' in the sense that it has a static, constraint-based formulation of Fitch's algorithm embedded within it. Although faster than naive brute force algorithms, the ILP for an unbounded number of states does not scale well and is limited to trees with approximately 16 taxa. On the other hand, the ILP for binary characters is fast: it can cope with trees with up to 100 taxa in reasonable time. In both cases the ILP is fast enough to verify the MP distance of a number of gadgets used in the hardness proofs. An implementation of this ILP has been made publicly available at \url{http://skelk.sdf-eu.org/mpdistbinary/} \citep{MPdistILP}.

\section{Notation}\label{sec:notation}
%
Recall that an {\it unrooted phylogenetic $X$-tree} is a tree $\T =(V(\T),E(\T))$ on a leaf set $X=\{1,\ldots,n\} \subset V(\T)$. Such a tree is named {\it binary} if it has only vertices of degree 1 (leaves) or 3 (internal vertices). A  {\it rooted phylogenetic $X$-tree} additionally has one vertex specified as the {\it root}, and such a rooted tree is named {\it binary} if the root has degree 2 and all other vertices are of degree 1 (leaves) or 3 (internal vertices). Note that two leaves are said to form a {\em cherry}, if they are connected to the same inner node. Moreover, a rooted binary tree on three taxa is also often referred to as a {\em rooted triplet}, and a rooted tree with only one cherry is also called a {\em caterpillar tree} or {\em caterpillar} for short. We often denote trees in the well-known Newick format \citep{felsenstein_2000}, which uses nested parentheses to group species together according to their degree of relatedness. For instance, the tree $((1,2),(3,4))$ is a tree with two so-called cherries $(1,2)$ and $(3,4)$ and a root between the two. 

Furthermore, recall that a {\it character} $f$ is a function $f: X\rightarrow \C$ for some set $\C:=\{c_1, c_2, c_3, \ldots, c_k \}$ of $k$ {\em character states} ($k \in \NN$). Often, $k$ is assumed to equal 4 in order for $\C$ to represent the DNA alphabet $\{A,C,G,T\}$, but in the present paper $k$ is not restricted this way but can be any natural number unless stated otherwise. Note that in the special case where $|f(X)|=2$, we also refer to $f$ as a binary character. In general, when $|f(X)|=r$, $f$ is called an \emph{$r$-state character}. In order to shorten the notation, it is customary to write for instance $f=AACC$ instead of $f(1)=A$, $f(2)=A$, $f(3)=C$ and $f(4)=C$. Note that each $r$-state character $f$ on taxon set $X$ partitions $X$ into $r$ non-empty and non-overlapping subsets $X_i$, $i=1,\ldots,r$, where $x_j,x_k \in X_i$ if and only if $f(x_j)=f(x_k)$. 

Note that in this paper, we refer to a character always with its underlying taxon set partition in mind, i.e. for instance we do not distinguish between $AACC$, $CCAA$ and $CCGG$, and so on. Moreover, when there is no ambiguity and when the stated result holds for both rooted and unrooted trees, we often just write `tree' or `phylogenetic tree' when referring to a phylogenetic $X$-tree.

An {\it extension} of a character $f$ to $V(\T)$ is a map $g: V(\T)\rightarrow \C$ such that $g(i) = f(i)$ for all $i$ in $X$. For such an extension $g$ of $f$, we denote by $l_{g}(\T)$ the number of edges $e=\{u,v\}$ in $\T$ on which a {\em substitution} occurs, i.e. where $g(u) \neq g(v)$. Such substitutions are also often referred to as {\em mutations} or {\em changes}. 
The {\em parsimony score} or {\em parsimony length} of a character $f$ on $\T$, denoted by $l_{f}(\T)$, is obtained by minimizing $l_{g}(\T)$ over all possible extensions $g$ of $f$. For binary trees $\T$, the parsimony score of a character $f$ can easily be calculated with the Fitch algorithm \citep{fitch_1971}. Recall that the bottom-up phase of Fitch starts at the labelled leaves and assigns to the unlabeled parent of two nodes the intersection of both children's label set if it is non-empty, or the union otherwise. The top-down phase then starts at the root with an arbitrary choice of the root states suggested by the bottom-up phase and keeps the current state for the descending nodes whenever this is contained in the label set of these nodes, and takes an arbitrary state out of the label set otherwise. 

This paper deals with the so-called {\em parsimony distance} $d_{MP}$ as introduced in \citep{fischer2014maximum}. This distance is defined as follows: Given two phylogenetic trees $\T_1$ and $\T_2$ on the same set $X$ of taxa, the parsimony distance between these trees is defined as $$d_{MP}(\T_1,\T_2)=\max_f |l_f(\T_1)-l_f(\T_2)|,$$ where the maximum is taken over all characters $f$ on $X$. A character $f$ which maximizes this distance is called an {\em optimal character}. Note that, due to the fact that the parsimony score of a tree (for a given character) is not affected by the presence or absence of a root, parsimony distance is also oblivious to whether the input trees are rooted or unrooted.

For some proofs in this paper we need the notion of a {\em maximum agreement forest}, which is closely linked to the so-called {\em rooted subtree prune and redraft distance} or {\em rSPR distance} for short. Recall that, informally, an agreement forest of two rooted phylogenetic trees is a set of subtrees which are identical in both trees and which in total contain all leaves, see, e.g. \citep{bordewichsemple2005}. A maximum agreement forest is an agreement forest with minimum number of components. A single rSPR move involves moving to a neighboring rooted tree by detaching a branch and re-attaching it elsewhere. The rSPR distance $d_{rSPR}$ is the minimum number of rSPR moves required to transform one rooted tree into another. Maximum agreement forests and  rSPR distance are closely linked by the well-known result that, modulo a rooting technicality, an agreement forest of two rooted trees with $m$ components yields a set of $m-1$ rSPR moves which turn the first tree into the second one \citep{bordewichsemple2005}. 

\section{Preliminaries}\label{sec:prelim}

The following observation, which we will use extensively and implicitly throughout the article, appeared unchanged in our earlier work \citep{fischer2014maximum}. 

\begin{observation}
\label{obs:mostone} Let $f$ be a character on $X$ and $\T$ a tree on $X$. Let $f'$ be any character obtained from $f$ by changing the state of exactly one taxon. Then
$l_f(\T) - 1 \leq l_{f'}(\T) \leq l_f(\T) + 1$ i.e. the parsimony score can change by at most one.
\end{observation}
\begin{proof}
Suppose $l_{f'}(\T) \leq l_{f}(\T) - 2$. Consider any extension of $f'$ to the interior
nodes of $\T$ that achieves $l_{f'}(\T)$ mutations. Using the same extension but on $f$
gives at most $l_{f'}(\T)+1$ mutations, because only one taxon changed state. So
$l_{f}(\T) \leq l_{f'}(\T)+1 \leq l_{f}(\T)-1$, which is a contradiction. In the other direction,
take any optimal extension of $f$ and apply it to $f$'. At most one extra mutation will be created, so $l_{f'}(\T) \leq l_{f}(\T) + 1$.
\end{proof}

A more general version of the following lemma appeared earlier in \citep{fischer2014maximum}. Here we have specialized the lemma and its proof to apply to rooted binary trees, which is the type of trees we will construct in the subsequent hardness reductions. 

\begin{lemma}
\label{lem:monochrome}
Let $f$ be an optimal character for two rooted, binary trees $\T_1$ and $\T_2$ i.e. $d_{MP}(\T_1, \T_2) = |l_f(\T_2) - l_f(\T_1)|.$
Without loss of generality assume $l_f(\T_1) < l_f(\T_2)$. Then we can construct in polynomial time an optimal character $f'$
with the
following property: $l_{f'}(\T_1) < l_{f'}(\T_2)$ and for each vertex $u$ of $\T_1$ such that both $u$'s children are leaves (i.e. they form a cherry), $f'$ assigns both children of $u$ the same state. 
\end{lemma}
\begin{proof}
Consider a vertex $u$ of $\T_1$ such that both of its children are taxa, but such that $f$ assigns
the two children different states. We calculate an optimal extension of $f$ to the interior
nodes of $\T_1$ by applying Fitch's algorithm. Let
$s$ be the state allocated to $u$ by the top-down phase of Fitch. Choose the child of $u$ that does not have state $s$ and
change its state to $s$. This yields a new character $f^{*}$. Clearly, $l_{f^{*}}(\T_1) < l_{f}(\T_1)$,
simply by using the same extension that Fitch gave. Combining this with Observation \ref{obs:mostone} gives
$l_{f^{*}}(\T_1) = l_{f}(\T_1) - 1$ and thus $l_{f^{*}}(\T_2) = l_{f}(\T_2)-1$ (otherwise $f$
could not have been optimal). Hence, $f^{*}$ is also an optimal character, and
$l_{f^{*}}(\T_1) < l_{f^{*}}(\T_2)$. This process can be repeated for as long as necessary. Termination in polynomial time is guaranteed because each taxon has its state changed at most once.
\end{proof}

\begin{observation}
\label{obs:boundedmonochrome}
Lemma \ref{lem:monochrome} also holds for optimal characters under the $d_{MP}^{i}(\T_1, \T_2)$ model.
\end{observation}
\begin{proof}
The transformation in the proof of Lemma \ref{lem:monochrome} does not increase the number of states in the
character.
\end{proof}


%

\section{MP distance on binary trees is NP-hard}

\subsection{The symmetry-breaking construction}
\label{subsec:symbreak}

In the hardness proof in Section \ref{subsec:reduc} we will construct two trees $\T_E$ and
$\T_V$ and a central fact used in the proof of correctness of the reduction is that, for all
optimal characters $f$, $l_f(\T_E) < l_f(T_V)$. In this section we show how to construct
a gadget to enforce this property. Note that all the trees constructed in this section
are binary. (As we demonstrated in \citep{fischer2014maximum} constructing such a symmetry-breaking gadget is far easier in the non-binary case).

\begin{figure}[ht]      \centering\vspace{0.5cm} 
    \includegraphics[width=10cm]{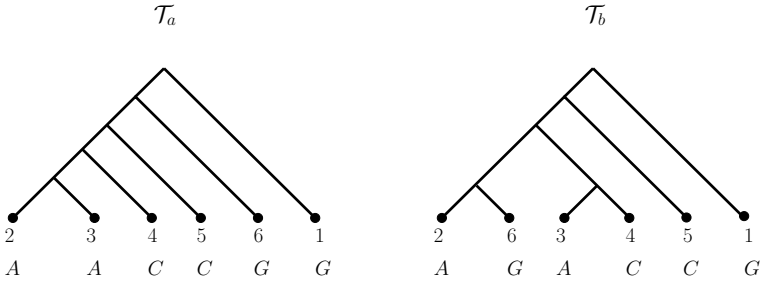}
 \caption{The two ``asymmetric'' trees $\T_a$ and $\T_b$ and an optimal character
$f_{asym}=GAACCG$. }
\label{fig:binaryAntisymunbounded}
\end{figure}

Consider the two rooted trees
\begin{align*}
\T_a = (((((2,3),4),5),6),1)\\
\T_b = ((((2,6),(3,4)),5),1)
\end{align*}
shown in Figure \ref{fig:binaryAntisymunbounded}. It can be verified computationally
that $d_{MP}(\T_a,\T_b) = 2$, achieved for example by
character $f_{asym}=GAACCG$\footnote{Note that for this specific character
there exist optimal extensions in both trees such that the root is allocated state $G$.} with $l_f(\T_a)=2$ and $l_f(\T_b)=4$.
Moreover,  if $f$ is an optimal character, then $l_f(\T_a) + 2 = l_f(\T_b)$. Expressed differently:
there does not exist any optimal character $f$ for which $l_f(\T_a) > l_f(\T_b)$, so the instance is ``asymmetric''. For two trees $\T_1$ and $\T_2$, let
\[
gap(\T_1, \T_2) =| \max_{f} ( l_f(\T_2) - l_f(\T_1)) - \max_{f}( l_f(\T_1) - l_f(\T_2))|
\] where $f$ ranges over all characters (not just optimal ones). Note that
 $gap(\T_a, \T_b) = 1$ because $\max_{f} ( l_f(\T_b) - l_f(\T_a)) = 2$
and $\max_{f}( l_f(\T_a) - l_f(\T_b)) = 1$, where e.g. the character $f = AACCAA$
achieves $l_f(\T_a) - l_f(\T_b) = 2-1= 1$.

We now describe an iterative construction such that, for any desired gap $g$, we can construct two trees $\T_1$ and $\T_2$, both on
$O(g)$ taxa, such that $gap(\T_1, \T_2) \geq g$.

We start with $\T_a$ and $\T_b$.  Let $\T_A$ be the rooted tree on 12 taxa obtained by taking two disjoint copies of $\T_a$ and joining them together via their roots $\rho_1, \rho_2$ to a newly introduced root $\rho$. (Here, the copying operation is assumed to introduce new taxon labels to prevent the same taxon occuring twice in the same tree). $\T_B$ is defined the same way, but with respect to $\T_b$.

\begin{claim}
$gap(\T_A, \T_B) \geq 2$.
\label{claim:gap2}
\end{claim}
\begin{proof}
We will show that $\max_f ( l_f(\T_B) - l_f(\T_A) ) \geq 4$ and
$\max_f( l_f(\T_A) - l_f(\T_B)) \leq 2$, from which the claim will follow. Let
$f$ be a character such that $l_f(\T_a) + 2 = l_f(\T_b)$, i.e. $f$ is an optimal character for
$\T_a, \T_b$. We extend character $f$
to become character $F$ on $\T_A, \T_B$ in the natural way i.e. disjoint copies of the
same taxon receive the same character state. If we run the bottom-up phase of Fitch's algorithm on $\T_A$ and $\T_B$, we observe that each copy of $\T_a$ induces 2 fewer mutations than its corresponding copy of $\T_b$. Moreover, the set of
states identified (by Fitch's bottom-up phase) to be possible at $\rho_1$ will be equal to the set of states identified
to be possible at $\rho_2$, so there will be no mutations incurred in $\T_A$ on the two edges
incident to its root $\rho$. By the same argument, there will be no mutations incurred in $\T_B$
on the edges incident to its root. Hence, $ l_F(\T_B) - l_F(\T_A) \geq 4$. Showing
$\max_f( l_f(\T_A) - l_f(\T_B)) \leq 2$ is possible analytically but it is technical. We
therefore omit the proof, noting however that we have used an exhaustive computational search to verify that (a) $\max_f ( l_f(\T_B) - l_f(\T_A) ) = 4$, where the maximum is reached e.g. by $f=ABBCCAACCBBA$ and
(b) $\max_f( l_f(\T_A) - l_f(\T_B)) = 2$, where the maximum is reached e.g. by $f=AABBAAAABBAA$. Note that our ILP described in Section \ref{sec:ilp} can also be used to verify the claim. The computational search thus allows us to draw the slightly stronger conclusion that $gap(\T_A, \T_B) = 2$.
\end{proof}

Let $\T^{k}_A$ be the rooted tree on $12k$ taxa obtained by arranging
$k$ disjoint copies of $\T_A$ along a caterpillar backbone. That is, $\T^{1}_A = \T_A$ and
for $k>1$, $\T^{k}_A$ is obtained by joining $\T^{k-1}_{A}$ and $\T_{A}$ via
a new root. $\T^{k}_B$ is defined analogously.

\begin{claim}
\label{claim:smallgap}
$gap(\T^{2}_A, \T^{2}_B) \geq 3$.
\end{claim}
\begin{proof}
By extending the character $F$ to $\T^{2}_A, \T^{2}_B$ in the usual fashion, and
using the same Fitch-based argument as in the previous proof, we
see that $\max_f ( l_f(\T^{2}_B) - l_f(\T^{2}_A) ) \geq 8$. One the other hand,
due to the fact that $\max_f( l_f(\T_A) - l_f(\T_B)) \leq 2$, the total number of mutations
incurred inside the two copies of $\T_A$ 
can in total be at most 4 more than the total number of
mutations incurred inside the two copies of $\T_B$. In the worst case, $\T^{2}_A$ can perhaps also suffer a single mutation on the two edges incident to the root, while $\T^{2}_B$ suffers none, so $\max_f ( l_f(\T^{2}_A) - l_f(\T^{2}_B) ) \leq 5$. The claim follows.
\end{proof}

\begin{lemma}
\label{lem:gap}
For $k\geq 1$, $gap(\T^{k}_A, \T^{k}_B) \geq k+1$.
\end{lemma}
\begin{proof}
We prove this statement by induction. For $k \in \{1,2\}$ the lemma has already been proved, so assume $k\geq 3$. By continuing the arguments used in the previous claims, we see that
\[
\max_f ( l_f(\T^{k}_B) - l_f(\T^{k}_A) ) \geq \max_f ( l_f(\T^{k-1}_B) - l_f(\T^{k-1}_A) ) + \max_f( l_f(\T_B) - l_f(\T_A))
\]
and
\[
\max_f ( l_f(\T^{k}_A) - l_f(\T^{k}_B) ) \leq \max_f ( l_f(\T^{k-1}_A) - l_f(\T^{k-1}_B) ) + \max_f( l_f(\T_A) - l_f(\T_B)) + 1
\]
where the 1 in the second expression accounts for the possibility that in $\T^{k}_A$
a mutation is incurred on one of the root edges, while no such mutation is incurred
in  $\T^{k}_B$. Combining the above with the fact that $\max_f( l_f(\T_B) - l_f(\T_A))=4$, $\max_f( l_f(\T_A) - l_f(\T_B)) = 2$, $\max_f ( l_f(\T^{2}_B) - l_f(\T^{2}_A) ) \geq 8$ and
$\max_f ( l_f(\T^{2}_A) - l_f(\T^{2}_B) ) \leq 5$,
we obtain the desired result.
\end{proof}

In addition to Lemma \ref{lem:gap}, we actually also need to know a (polynomial-time computable) expression for $d_{MP}(\T^{k}_A, \T^{k}_B)$. Conveniently, we have
a closed expression for this.

\begin{lemma}
\label{lem:gapdist}
For $k \geq 2$, $d_{MP}(\T^{k}_A, \T^{k}_B) = 8 + 4(k-2) = 4k$.
\end{lemma}
\begin{proof}
From the proof of Lemma \ref{lem:gap} we know that
\[
d_{MP}(\T^{k}_A, \T^{k}_B) = max_f ( l_f(\T^{k}_B) - l_f(\T^{k}_A) ).
\]
Due to the recurrence shown in the proof of that lemma we see,
\[
max_f ( l_f(\T^{k}_B) - l_f(\T^{k}_A) ) \geq 8 + 4(k-2)
\]
We will complete the proof by showing
$d_{MP}(\T^{k}_A, \T^{k}_B) \leq 8 + 4(k-2)$.  To do this, we exploit the fact (proven in \citep{fischer2014maximum}) 
that $d_{MP}(\T^{k}_A, \T^{k}_B) \leq d_{rSPR}(\T^{k}_A, \T^{k}_B)$ i.e. MP distance is a lower bound on the well-known \emph{rooted subtree prune and regraft} (rSPR) distance. In particular,
we prove that $d_{rSPR}(\T^{k}_A, \T^{k}_B) \leq 8 + 4(k-2)$. We do this by showing
that $\T^{k}_A, \T^{k}_B$ permit an \emph{agreement forest} with at most
$8 + 4(k-2) + 1 = 4k + 1$ components. (It is well-known that an agreement forest with $m$ components yields a set of $m-1$ rSPR moves that turn one tree into the other, see \citep{bordewichsemple2005})\footnote{To utilize this agreement forest formulation of rSPR we should first append a new taxon $\rho$ to the root of both trees. However in this case it is easy to check that the omission of $\rho$ does not harm the analysis.}. Now,
observe that $\T^{k}_A, \T^{k}_B$ contains $4+2(k-2) = 2k$ copies of the original $\T_a, \T_b$
trees. Next, observe that an agreement forest for $\T_a, \T_b$ with 3 components can be obtained by placing taxon $3$ and taxon $6$ each in a singleton component, and $\{1,2,4,5\}$
in the remaining component. To obtain an agreement forest for $\T^{k}_A, \T^{k}_B$ we
put all copies of taxon $3$ and all copies of taxon $6$ in singleton components, yielding
$4k$ singleton components. All remaining taxa can be placed in one large component, yielding
$4k+1$ components in total.  
\end{proof}

Finally, we consider the following auxiliary observation, which will be useful later.

\begin{observation}
\label{obs:sameroot}
For each $k \geq 2$, there exists an optimal character $f^k$ on $\T^{k}_A, \T^{k}_B$ such
that $f^{k}$ has 3 states, and there exist optimal extensions of $f^k$ to both trees, such that the roots of  $\T^{k}_A, \T^{k}_B$ both receive the same state.
\end{observation}
\begin{proof}
As noted earlier,  $f_{asym}=GAACCG$ is an optimal character for $\T_a, \T_b$ and
permits optimal extensions such that the roots of both trees can be assigned state $G$.
We can obtain an optimal character $f^k$ on $\T^{k}_A, \T^{k}_B$ simply by making
$2k$ copies of $f_{asym}$. The optimality of $f^{k}$ follows from the fact that
$f_{asym}$ is optimal for $\T_a, \T_b$ and that Claim \ref{claim:gap2} holds for any optimal character. Given that each copy of $\T_a$ and $\T_b$ can have state $G$ allocated to its root, it
follows (by continuing the bottom-up phase of Fitch on the remainder of $\T^{k}_A$ and  $\T^{k}_B$) that there exist optimal extensions of $f^k$ such that the roots of $\T^{k}_A$ and $\T^{k}_B$ are both allocated state $G$. 
\end{proof}

\subsection{The reduction}
\label{subsec:reduc}

In this section we exclusively consider simple undirected graphs.
Recall that a graph $G=(V,E)$ is \emph{cubic} if every vertex has degree exactly 3, in which case
$|E| = 3|V|/2$. A \emph{proper edge colouring} of a graph $G$ is an assignment of colours to the edges such that no two adjacent edges have the same colour, where two edges are adjacent if they have a common endpoint. Let $\chi'(G)$, the \emph{chromatic index} of $G$, be the minimum number of colours required
to properly colour the edges of $G$. The classical result of Vizing (see any standard graph-theory text, such as \citep{diestel2005}) states that for every graph $G$, $\Delta(G) \leq \chi'(G) \leq \Delta(G) + 1$ where $\Delta(G)$ is the maximum degree of
a vertex in $G$. Hence, for cubic $G$, $\chi'(G) \in \{3, 4\}$. Even for
cubic graphs it is NP-hard to distinguish between these two possibilities \citep{holyer1981np}. 

\begin{theorem}
\label{thm:hard}
Computation of $d_{MP}(\T_1, \T_2)$ on binary trees is NP-hard.
\end{theorem}
\begin{proof}
Let $G = (V,E)$ be a cubic graph where $n=|V|$. We give a polynomial-time reduction from computation
of $\chi'(G)$ to computation of $d_{MP}$, from which NP-hardness will follow. Specifically,
we will construct two trees $\T_E$ and $\T_V$ such that, for a certain integer $P$, $d_{MP}(\T_E, \T_V) = P$ if and only if $\chi'(G) = 3$. In particular, if $\chi'(G) = 4$, then $d_{MP}(\T_E, \T_V)$ will be $P-1$ (or less). An important difference with
Theorem 4.6 of \citep{fischer2014maximum} is that here optimal characters $f$ will be engineered to always have
the property $l_f(\T_E) < l_f(\T_V)$ and not the
other way round. Informally, at optimality $\T_E$ always ``wins''.

The high-level idea is that in $\T_E$ we will choose the colours of the edges of $G$. In fact,
for each edge we will choose three colours, all different, representing the colour of $e$ in three different copies of $G$. Due
to the way we construct the two trees, there will exist optimal characters in which
the edge colouring (in each of the three copies of $G$) is \emph{proper}. This is because, the closer an edge colouring is to being
proper, the higher the parsimony score induced in $\T_V$. Within the space of proper
edge colourings, we will show that it is advantageous to use as few colours as possible, because
this will give the character a low parsimony score on $\T_E$. Leveraging the fact that
the colours used for the three copies of each edge are all different, we will derive the conclusion
that $d_{MP}$ can reach a certain value $P$ if and only if there is a proper edge colouring that uses only 3 colours i.e. that $\chi'(G)=3$.
We will prove the following:

\begin{align*}
&\chi'(G)=3 \Rightarrow d_{MP}(\T_E, \T_V) = P\\
&\chi'(G)=4 \Rightarrow d_{MP}(\T_E, \T_V) \leq P - 1
\end{align*}

Let $M$ be a large integer, at most polynomially large in $n$, whose value we will specify later. Letting $k=M$, construct $\T^{k}_{A}, \T^{k}_{B}$ (as described in the previous section). Relabel $S_E = \T^{k}_{A}$ and $S_V = \T^{k}_{B}$. By Lemma \ref{lem:gap}, $gap(S_E, S_V) \geq M+1$.

The core ingredients of $\T_E$ are the subtrees $B$, $S_E$ and $\T^{***}$. We construct $B$ by taking an arbitrary rooted binary tree on $3|V|+|E|$ taxa. By appending an extra taxon $\alpha$ just above its root, we create a new root yielding $3|V|+|E| + 1$ taxa in total. Note that since $\alpha$ is not a taxon of $B$, in the following we refer to $B$ including $\alpha$ or $B$ without $\alpha$ to stress whether or not $\alpha$ is considered together with $B$ or not.

Tree $T^{***}$ is constructed as follows. Fix an arbitrary rooted binary tree $T^{*}$ on $|E|$ leaves, identifying the leaves with elements of $E$. Replace
each leaf $u_e$ of $T^{*}$, where $e \in E$, with a rooted triplet to obtain $T^{**}$ on $3|E|$ leaves $u_{e,j}$ where $e \in E$ and $j \in \{1,2,3\}$. Finally, replace each leaf $u_{e,j}$  of $T^{**}$ with a rooted triplet on three taxa $x_{e,j}^{u}, x_{e,j}^{v}$ and $x_{e,j}^{edge}$ where $u, v \in V$
are the two endpoints of $e$. We ensure that $x_{e,j}^{u}, x_{e,j}^{v}$ are sibling to each other (i.e. form a cherry). This is $T^{***}$, which is depicted in Figure \ref{fig:Ttriplestar}, and it has thus $9|E|$ taxa.

\begin{figure}[ht] 
\centering\vspace{0.5cm}
   \scalebox{.8}{ \includegraphics[width=14cm]{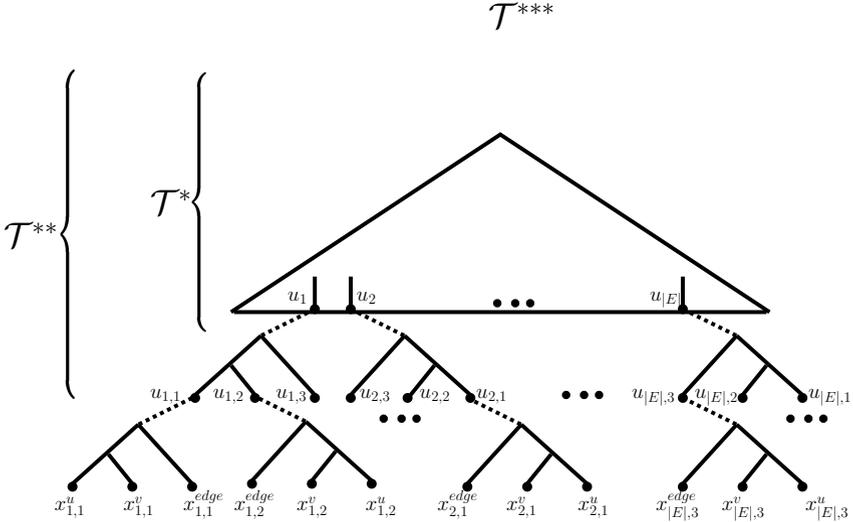}}
\caption{The tree $T^{***}$. Here we have identified $E$ with the set $\{1,\ldots,|E|\}$ to simplify the figure. In the lowermost leaves
we have overloaded $u$ and $v$: in each case they refer to the two endpoints of the edge in question. } \label{fig:Ttriplestar}
\end{figure}

The basic idea is that each edge $e = \{u,v\}$ occurs $3$ times in total, and each such occurrence consists of a cherry representing $u$ and $v$, and an extra taxon (``\emph{edge}'' ) sitting just above the cherry.

The construction of $\T_E$ is concluded by joining $B$ including $\alpha$, $S_E$ and $T^{***}$
 as shown
in Figure \ref{fig:binaryTEbeep}, which also introduces auxiliary taxa $\beta_1, \beta_2, \gamma_1, \gamma_2$. We adopt the labels used in that figure. Summarizing,
$\T_E$ contains
\[
 3|V|+|E| + 1 + 4 + 9|E| + 12M
\]
taxa.

\begin{figure}[ht] 
\centering\vspace{0.5cm}
    \scalebox{.8}{ \includegraphics[width=14cm]{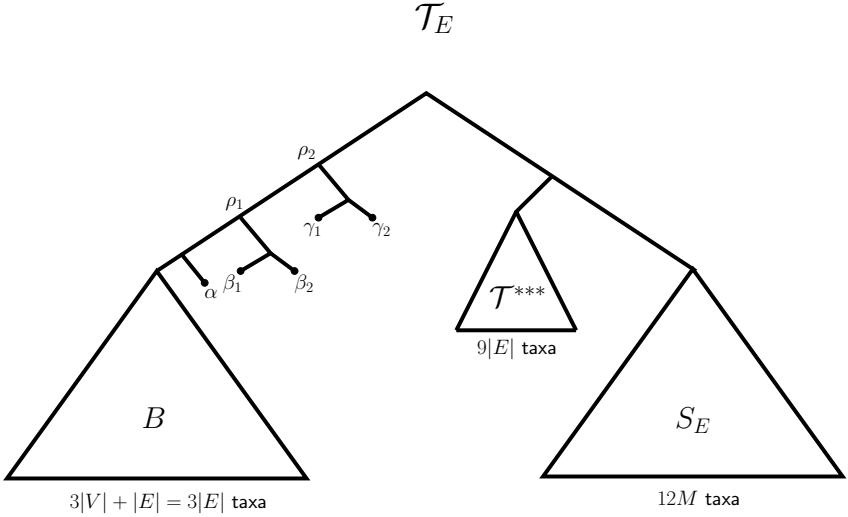}}
\caption{ The tree $\T_E$. Taxon $\alpha$ is closely linked to subtree $B$ as it is descending from the same root as $B$ in $\T_E$. This root is considered in the proof. However, in $\T_V$, $\alpha$ shares a direct common root with $S_V$, not $B$. }\label{fig:binaryTEbeep}
\end{figure}

To construct $\T_V$ we start by taking $B$ and attaching $S_V$ on the edge entering
taxon $\alpha$. Now, let 
\[
H = \{ (v, j) |  v \in V, j \in \{1,2,3\} \} \cup
 \{ e |  e \in E \}.
\]
Clearly, $|H| = 3|V|+|E|$. Pick an arbitrary bijection between the taxa of $B$ (excluding $\alpha$) and the elements of $H$. For each edge $e \in H$, introduce a rooted triplet
on the three taxa $x_{e,1}^{edge}$, $x_{e,2}^{edge}$, $x_{e,3}^{edge}$ and attach this rooted
triplet on the edge entering the taxon of $B$ corresponding to $e$.  For each
tuple $(v,j) \in H$, let $\{e, e^{*},e^{**}\}$ be the 3 edges incident to $v$ in $G$,
introduce a rooted triplet on the three taxa $x_{e,j}^{v}, x_{e^{*},j}^{v}$
and $x_{e^{**},j}^{v}$, and attach this rooted triplet on the edge entering the
taxon of $B$ corresponding to $(v,j)$.  Finally, we introduce a new root and join
$B$ to the new subtree on $((\beta_1, \gamma_1),(\beta_2, \gamma_2))$.
This completes the construction of $\T_V$, which is depicted in Figure \ref{fig:TVunbounded}.

\begin{figure}[ht] 
\centering
    \scalebox{.8}{ \includegraphics[width=14cm]{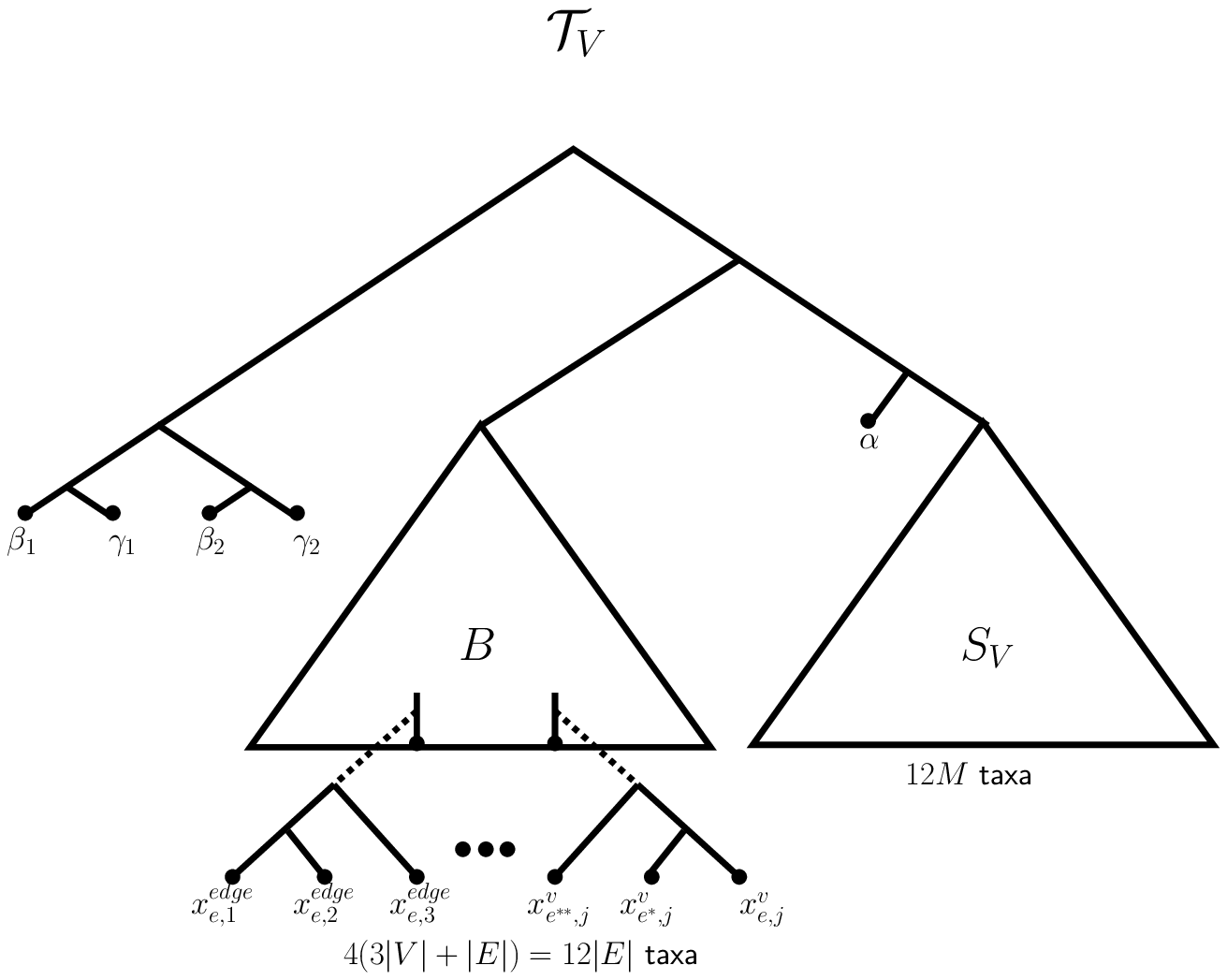}}
\caption{ The tree $\T_V$.  }\label{fig:TVunbounded}
\end{figure}
 
We are now in a position to specify the number $M$. We require $M$ to be sufficiently
large that, for every optimal character $f$, $l_f( \T_E ) < l_f( \T_V )$. From
Lemma \ref{lem:gapdist} we know that there exists some character $f'$ such
that $l_{f'}( \T_V ) - l_{f'}(\T_E) \geq 4M$. (In particular, we can obtain such a character
by -- for example -- extending the character suggested by Lemma \ref{lem:gapdist} such that all
taxa outside $S_E$ and $S_V$ are assigned the same state.) Now, let $t$ be the number
of edges in $\T_E$ that lie outside $S_E$. For every character $f$ we have
\[
l_f( \T_E ) - l_f(\T_V) \leq t + (4M - (M+1)).
\]
The $4M$ term is obtained from Lemma \ref{lem:gapdist}, the $(M+1)$ term from Lemma \ref{lem:gap}, and the $t$ term arises (pessimistically) from the situation when every
edge in $\T_E$ (outside $S_E$) incurs a mutation, but no edge in $\T_V$ (outside $S_V$)
incurs a mutation. So, if we choose $M$ such that
\[
t + 4M - (M+1) < 4M
\]
it follows that for \emph{every} optimal character $f$, $d_{MP}(\T_V, \T_E) = l_f(\T_V) - l_f(\T_E)$ and in particular $l_f(\T_V) > l_f(\T_E)$. Choosing $M = t$ is therefore sufficient
to achieve this. This ``symmetry breaking'' has far-reaching consequences which we shall heavily utilize later.

Next, let $f^{S}$ be any 3-state character on the taxa in $S_E$ and $S_V$ such that
$l_{f^S}(\T_V) - l_{f^S}(\T_E) = 4M = d_{MP}(S_E, S_V)$. This character exists and can be constructed in polynomial time thanks to Observation \ref{obs:sameroot}. Recall, in particular, that it is constructed by making many disjoint copies of the character $f_{asym} = GAACCG$.

Now, suppose $\chi'(G) = 3$. We will extend $f^{S}$ to all the
taxa in $\T_E$ as follows, obtaining a 4-state character. Take any proper edge colouring $Col$ of graph $G$ using three colours \emph{red, blue} and \emph{green}. We start by relabelling $f^{S}$ as follows: character state $G$ maps to \emph{blue},
$A$ maps to \emph{red} and $C$ maps to \emph{green}. Next, colour all the taxa in $B$ including $\alpha$
\emph{pink}. Colour the cherry $\{ \beta_1, \beta_2 \}$ \emph{pink} and the
cherry $\{ \gamma_1, \gamma_2 \}$ \emph{blue}.

 Next, consider the following cyclical mapping $F$:
\begin{align*}
F(red,1) &\rightarrow red\\
F(red,2) &\rightarrow blue\\
F(red,3) &\rightarrow green\\
F(blue,1) &\rightarrow blue\\
F(blue,2) &\rightarrow green\\
F(blue,3) &\rightarrow red\\
F(green,1) &\rightarrow green\\
F(green,2) &\rightarrow red\\
F(green,3) &\rightarrow blue\\
\end{align*}
Now, for every $e \in E$, $j \in \{1,2,3\}$
and letting $e = \{u,v\}$, we assign 
$x_{e,j}^{u}, x_{e,j}^{v}$ and $x_{e,j}^{edge}$ all the same colour: the colour
$F( Col(e), j)$ where as usual $Col(e)$ is the colour assigned to $e$ by the proper edge colouring $Col$. 

Let this character be called $f^{Col}$.  Observe,
\[
l_{f^{Col}}(\T_E) =  1 + 2|E| + l_{f^{S}}(S_E)
\]
This can be confirmed by applying Fitch's algorithm. Note, in particular, that there is an optimal
extension such that all the internal nodes of the tree $B$ (including $\alpha$) are coloured \emph{pink},  all the
nodes of the $T^{*}$ part of $T^{***}$ are \emph{blue}, $p_2$ is \emph{blue}, the root is \emph{blue}, and all the unlabelled nodes are \emph{blue}. The $+1$
is then the mutation in the transition from \emph{pink} to \emph{blue} on, for example,
the edge between the cherries $\{ \beta_1, \beta_2 \}$ and $\{ \gamma_1, \gamma_2 \}$.
There is no
mutation on the edge entering the root of $S_E$ because, by Observation  \ref{obs:sameroot}
and the way we relabelled $f^{S}$, there is an optimal extension of $S_E$ in which its
root is permitted to be \emph{blue}.

\noindent
Also,
\[
l_{f^{Col}}(\T_V) =  2 + 3(3|V|+|E|) + 1 +  l_{f^{S}}(S_V)
\]
The $+1$ term here is definitely incurred because there is an optimal extension in which the root of $\T_V$ and $\alpha$ are both coloured \emph{pink}, but \emph{pink} is not used in $f^{S}$,
so there will then be a mutation on the edge entering the root of $S_V$. The $+2$ term
corresponds to mutations incurred in the $\beta_1, \beta_2, \gamma_1, \gamma_2$ region of $\T_V$.

Now, define $P$ as follows:
\begin{align*}
P & = l_{f^{Col}}(\T_V) - l_{f^{Col}}(\T_E)\\
&= \bigg ( 2 + 3(3|V|+|E|) + 1 +  l_{f^{S}}(S_V) \bigg ) - \bigg ( 1 + 2|E| + l_{f^{S}}(S_E) \bigg )\\
& = \bigg ( 2 + 3(3|V|+|E|) + 1  \bigg ) - \bigg ( 1 + 2|E|  \bigg ) + 4M
\end{align*}

Hence, if $\chi'(G)=3$, $d_{MP}(\T_E, \T_V) \geq P$. We still need to show
(1) $d_{MP}(\T_E, \T_V) \leq P$ and (2) $d_{MP}(\T_E, \T_V) = P$ if and only if $\chi'(G)=3$. Once these facts have been established NP-hardness will follow.

We approach this by starting from an arbitrary optimal character $f$ and then transforming
$f$ step by step such that we do not lose optimality but the character attains a certain canonical form. This canonical form will be attained by accumulating one special property at a time. In all
cases the argument that a new property can be obtained, is based on the assumption
that all earlier properties have already been accumulated. Properties are never lost, and each property can
be attained in polynomial time. Thus, given an arbitrary optimal character we can transform it in polynomial time
into a character that has all the described properties. The proofs of the following properties can be found in the Appendix (unless stated here).

\noindent
\emph{\textbf{Property 1.} All cherries in $\T_E$ are monochromatic.
That is, if $\{x,y\}$ 
are two taxa that share a parent in $\T_E$, then both are assigned the same colour (i.e. state).}\\
\\
\noindent
\emph{Proof.} This is possible by combining Lemma \ref{lem:monochrome} with the fact (established earlier)
that, for every optimal character $f$, $l_f(\T_E) < l_f(\T_V)$.\\
\\
\noindent
\emph{\textbf{Property 2.1.} In $\T_E$, the cherry $\{ \beta_1, \beta_2 \}$ has a different
colour to the cherry $\{ \gamma_1, \gamma_2 \}$.}\\
\\
\noindent
\noindent
\emph{\textbf{Property 2.2.} In $\T_E$, the (possibly multiple) colours used for the taxa of $B$ (including $\alpha$) are not used
elsewhere in $\T_E$, except possibly $\{ \beta_1, \beta_2 \}$. }\\
\\
\noindent
\emph{\textbf{Property 2.3} In $\T_E$, all the taxa in $B$ have the same colour which, with
the possible exception of $\beta_1,\beta_2$, does not appear on taxa outside $B$ and $\alpha$.}\\
\\
\emph{\textbf{Property 3.} In $\T_E$, all the taxa in $B$ have the same colour, and cherry
$\{\beta_1, \beta_2\}$ also has this colour. Moreover this colour does not appear on any
other taxa i.e. it is unique for $B$ (including $\alpha$) and $\beta_1, \beta_2$.}\\
\\
\noindent
From now on we refer to the unique colour used by $B$ (including $\alpha$),
$\beta_1$ and $\beta_2$ as \emph{pink}. Property 3 is extremely important. 
In particular, it means that from now on we can assume the existence of optimal extensions of $\T_V$ such that the root of $\T_V$ is coloured \emph{pink} and, moreover, that the entire image of $B$ inside $\T_V$ is coloured \emph{pink}. We call these \emph{pink extensions}. These greatly simplifies the task of counting mutations inside $\T_V$. In particular, it means that we from now on (in $\T_V$) only need to consider mutations incurred \emph{inside} the subtrees sibling to the taxa of $B$, which we call \emph{below pink} subtrees. These subtrees
never contain \emph{pink} taxa.\\

\emph{\textbf{Property 4}. Let $f$ be an optimal character with properties 1--3 and let $f^{*}$ be the
restriction of $f$ to the taxa in $S_V$ and $S_E$. Then $l_{f^*}(S_V) - l_{f^*}(S_E) = d_{MP}(S_E, S_V)$.}\\
\\
\noindent
\emph{Proof. } Fix a \emph{pink extension} of $f$. From the earlier properties, $f^{*}$ does not contain any \emph{pink} taxa. Now, taxon $\alpha$ is coloured \emph{pink}, because $\alpha$ is a taxon of $B$. This means that, in $\T_V$, there is unavoidably a mutation on the edge entering the
root of $S_V$.  Moreover, we know that there exist optimal characters for $S_E, S_V$ in which the roots of $S_E$ and $S_V$ can be allocated the same colour \emph{blue}: this is the 3-state character $f^{S}$ that we constructed at the start of the proof.  This means that, without loss of
optimality, we can assume $f^{*} = f^{S}$, where we are free to (and should) relabel the \emph{blue} inside
$f^{S}$ such that in $\T_E$ \emph{no} mutation is incurred on the edge entering the root
of $S_E$. (This can be achieved by running the bottom-up phase of Fitch on the subtree
sibling to $S_E$ in $\T_E$, identifying the set of colours permitted by Fitch at the root of the subtree, and arbitrarily picking one of those colours as the relabelling colour). Optimality is
assured because (1) $l_{f^S}(S_V) - l_{f^S}(S_E) = d_{MP}(S_E, S_V)$,  (2) we force
a mutation at the root of $S_V$ and (3) we definitely avoid a mutation at the root of $S_E$.\\
\\
\noindent
\emph{\textbf{Property 5}. (a) For every edge $e = \{u,v\} \in E$  the
three taxa $x_{e,1}^{edge}$, $x_{e,2}^{edge}$, $x_{e,3}^{edge}$ all have distinct colours. Moreover, (b) $x_{e,1}^{edge}, x_{e,1}^{u}, x_{e,1}^{v}$ all
have the same colour, $x_{e,2}^{edge}, x_{e,2}^{u}, x_{e,2}^{v}$ all
have the same colour, and finally $x_{e,3}^{edge}, x_{e,3}^{u}, x_{e,3}^{v}$ all
have the same colour.}\\
\\
\noindent
\emph{Proof.} First, suppose for some $e \in E$ there exists
$j, j' \in \{ 1,2,3\}$ such that $j \neq j'$ and $x_{e,j}^{edge}$, $x_{e,j'}^{edge}$
have the same colour. Observe that  $x_{e,1}^{edge}$, $x_{e,2}^{edge}$, $x_{e,3}^{edge}$ all form a single \emph{below pink} subtree in $\T_V$. Suppose
we recolour $x_{e,j}^{edge}$ to some brand new colour. This raises the parsimony
score of $\T_E$ by at most 1. However, it also raises the parsimony score  of $\T_V$ by
at least one, due to the introduction of a new colour into the corresponding \emph{below pink}
subtree. Hence the recoloured character is optimal. We can repeat this as long as necessary
to ensure that (a) eventually holds.  Now, suppose for some $e \in E$  and $j \in \{1, 2, 3\}$ the taxa $x_{e,j}^{edge}, x_{e,j}^{u}, x_{e,j}^{v}$ do not all have the same colour. By Property 1 we know that $x_{e,j}^{u}$ and $x_{e,j}^{v}$ have
the same colour, because they form a cherry in $\T_E$.
We recolour all 3 taxa with a 
brand new colour. This cannot raise the parsimony score of $\T_E$. On
the other hand, it cannot lower the parsimony score of $\T_V$, because
the three now uniquely coloured taxa all occur in different \emph{below pink} subtrees of $\T_V$. Hence the recoloured character is optimal, and (a) is still holding. We repeat this as long
as necessary to ensure that (b) eventually also holds.\\
\\
\noindent
\emph{\textbf{Property 6}. For every $j \in \{1, 2, 3\}$,
the edge colouring induced by the colours  of the $x_{e,j}^{edge}$ taxa ($e \in E$), is
a proper edge colouring.}\\
\\
Recall that, by Property 5, each $x_{e,j}^{edge}$ taxon has the same colour as the $x_{e,j}^{u}$ and $x_{e,j}^{v}$ taxa below it in $\T_E$. Suppose that there is some  $j \in \{1,2,3\}$ for which the induced edge colouring is not proper.
Then there exists some $u \in V$ and two
edges $e \neq e'$ in $E$ incident at $u$ such that $x_{e,j}^{u}$
and $x_{e',j}^{u}$ both have the same colour. Both these taxa are together
in a \emph{below pink} subtree of $\T_V$. This subtree therefore currently induces $m<2$ mutations (excluding the mutation as the subtree touches the \emph{pink} region). Now,
suppose we introduce a brand new colour and recolour $x_{e,j}^{u}, x_{e,j}^{v}$ and $x_{e,j}^{edge}$ with it. This raises the parsimony score of $\T_E$ by at most 1. However,
it definitely also raises the parsimony score of $\T_V$, by at least 1, because the aforementioned
\emph{below pink} subtree now induces $m+1$ mutations (due to the introduction of a new colour). Hence the new character is optimal, and all earlier properties are preserved. We can repeat this process until the induced edge colouring is proper.\\
\\
\noindent
\emph{\textbf{Property 7}. For an optimal character $f$,}
\[
l_f(\T_V) = 2 + 3(3|V|+|E|) + 1 +  l_{f^{S}}(S_V).
\]
\emph{Proof. } This is a consequence of the fact that (from Property 6) we can assume
that in $\T_V$ a proper edge colouring is induced, plus the fact that a \emph{pink extension} is 
an optimal extension. In particular, the proper edge colouring means that each
of the $3|V|+|E|$ \emph{below pink} subtrees induces 2 mutations on its internal
edges and a third mutation where the subtree touches the \emph{pink} region. The `2'
term corresponds to the fact that the two taxa $\gamma_1, \gamma_2$ are necessarily
not pink. The  '1' term is the mutation at the root of $S_V$.\\
\\
\noindent
\textbf{Central argument}\\
\\
\noindent
As a consequence of Property 7, optimal characters (which we always assume to induce proper edge colourings) are only distinguished by their ability to minimize the number of mutations induced in $\T_E$. We can already establish a strong lower bound for this number:
\[
l_f( \T_E ) \geq 1 + 2|E| + l_{f^{S}}(S_E)
\]
Every proper edge colouring induces (at least) these mutations in $\T_E$. The '1' term
is the mutation that occurs between the $\beta_1, \beta_2$ and $\gamma_1, \gamma_2$
taxa and the $2|E|$ term is a consequence of (amongst others) Property 5.

Hence, $d_{MP}(\T_E, \T_V) = l_f(\T_V) - l_f(\T_E) \leq P$, where $P$ is the value defined earlier in the proof.
We have already shown that, if $G$ has $\chi'(G)=3$, $P$ is possible. We now see
that this is optimal. The only thing we have left to show, is that if $\chi'(G) > 3$, that
$P$ is \emph{not} possible. We use the contrapositive to prove this. In particular, we
will show
\[
l_f( \T_E ) = 1 + 2|E| + l_{f^{S}}(S_E) \Rightarrow \chi'(G)=3.
\]
Suppose, then, that $l_f( \T_E ) = 1 + 2|E| + l_{f^{S}}(S_E)$. 
This means that there are no mutations in the subtree $T^{***}$ other than the $2|E|$ unavoidable mutations due to Property 5. To achieve this
it must be the case that all the $|E|$ subtrees (each containing 9 taxa) in $T^{***}$ all have a single
colour in common. Let us call this colour \emph{blue}. Hence, for every $e \in E$,
there exists exactly one $j \in \{1, 2, 3\}$ such that 
 $x_{e,j}^{edge}, x_{e,j}^{u}, x_{e,j}^{v}$ are all \emph{blue}. We now
build a proper 3-edge-colouring for $G$. If $j=1$, we assign $e$ the colour \emph{red}.
If $j=2$, we assign $e$ the colour \emph{blue}. If $j=3$, we assign $e$ the colour
\emph{green}. This must be a proper colouring: if it was not, then there would be some
vertex $u \in V$, two incident edges $e, e'$ incident to $u$, and some $j \in \{1,2,3\}$
such that  $x_{e,j}^{edge}$ and $x_{e',j}^{edge}$ were both \emph{blue}. But
this would contradict Property 6. Hence, $\chi'(G)=3$.\\
\\
This completes the proof. Summarising, for a given cubic graph $G=(V,E)$, $$\chi'(G)=3 \Leftrightarrow 
d_{MP}(\T_E, \T_V) = P,$$
from which the NP-hardness of computing $d_{MP}(\T_E, \T_V)$ on binary trees follows.

\end{proof}

\begin{corollary}
\label{cor:fixed}
For every fixed integer $i \geq 4$, computation of $d^{i}_{MP}$ on binary trees is
NP-hard.
\end{corollary}
\begin{proof}
This is a consequence of the fact that in the theorem only 4 states are required to construct
a character achieving MP distance $P$. Namely, the 3 colours used in the proper edge colouring
of $G$, plus \emph{pink}.
\end{proof}

Note that the above proof cannot (obviously) be extended to give APX-hardness. By taking multiple copies of the tree $T^{***}$ it is possible to increase the gap between $\chi'(G)=3$ and $\chi'(G)=4$ instances to more than 1, but this is insufficient for APX-hardness.


\newpage

\section{Computation of $d^{2}_{MP}$ is NP-hard on binary trees.}

As in the previous section we first require a gadget that can break symmetry between
two trees.

\subsection{Symmetry breaking gadget in the case of 2 states}
\label{subsec:2statesymmetry}

Consider the two rooted trees
\[
T_a = (((5,(6,4)),3),((1,(8,2)),7))
\]
and
\[
T_b = (((7,((4,2),6)),3),(8,(1,5)))
\]

shown in Figure \ref{fig:binaryAntisym2states}.

\begin{figure}[ht]      \centering\vspace{0.5cm} 
      \includegraphics[width=10cm]{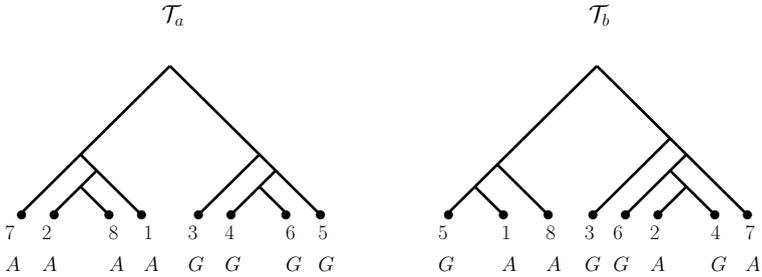}
\caption{Two trees $\T_a$and $\T_b$ that are ``asymmetric'' on
characters with at most 2 states. An example of an optimal character
is $f_{asym} = AAGGGGAA$. }
\label{fig:binaryAntisym2states}
\end{figure}

Here, it can be verified (e.g. by exhaustive search) that $d^{2}_{MP}(\T_a, \T_b)=3$, and the character $f_{asym} = AAGGGGAA$
can achieve this: $l_{f_{asym}}(\T_a)=1$ and $l_{f_{asym}}(\T_b)=4$.


In fact,
these trees are asymmetric, in the sense that for every optimal 2-state character $f$,
$l_f( \T_a ) < l_f( \T_b )$. In particular,
as can be verified by computational search (e.g. using the ILP formulation or performing an exhaustive search),
$\max_{f} ( l_f(\T_b) - l_f(\T_a)) = 3$ and $\max_{f} ( l_f(\T_a) - l_f(\T_b))=2$. 
(The second maximum is achieved by the character $AAGAAGGG$, for example.)
Using the same notation as in Section \ref{subsec:symbreak}, but
restricted to characters with at most 2 states, we therefore obtain:
\[
gap(\T_a, \T_b) = 1.
\]
From now on we implicitly assume that all characters have
at most 2 states.

Define $\T_A$ and $\T_B$ in the same way as in Section \ref{subsec:symbreak}. It can be verified that
$gap(\T_A, \T_B) \geq 1$. This is not yet strong enough for what we require, so let $\T_{AA}$
and $\T_{BB}$ be obtained by joining two copies of $\T_A$, and two copies of $\T_B$,
together (respectively). 

\begin{claim}
\label{clm:AA}
$gap(\T_{AA}, \T_{BB}) \geq 2.$
\end{claim}
\begin{proof}
It can easily be checked that $\max_{f}( l_f(\T_{BB}) - l_f(\T_{AA})) \geq 12$. This can
be achieved, for example, by taking a character $f$ that comprises 4 disjoint copies
of $f_{asym}$, thus obtaining $l_f(\T_{BB}) = 16$ and $l_f(\T_{AA})=4$. (In fact, by performing an exhaustive search, one can show that this is optimal).  Verifying
that $\max_{f}( l_f(\T_{AA}) - l_f(\T_{BB})) = 10$ is more challenging. We have used an exhaustive search to check this, but note that our ILP gives the same result in significantly less time. In fact, $l_f(\T_{AA}) - l_f(\T_{BB}) = 10$ can be achieved by $f=AGAGAGGAAGGGAGGAGAAAGAAGAGGGAGGA$, for which the score on $\T_{AA}$ is 14 and the score on tree $\T_{BB}$ is 4. So, altogether we have
$\max_{f}( l_f(\T_{BB}) - l_f(\T_{AA})) = 12$ and
$\max_{f}( l_f(\T_{AA}) - l_f(\T_{BB})) = 10$, so $gap(\T_{AA}, \T_{BB}) \geq 2$ and
$d^{2}_{MP}(\T_{AA}, \T_{BB})=12$.
\end{proof}

Let $\T^{k}_{AA}$ be the rooted tree on $32k$ taxa obtained by arranging
$k$ disjoint copies of $\T_{AA}$ along a caterpillar backbone. That is, $\T^{1}_{AA} = \T_{AA}$ and
for $k>1$, $\T^{k}_{AA}$ is obtained by joining $\T^{k-1}_{AA}$ and $\T_{AA}$ via
a new root. $\T^{k}_{BB}$ is defined analogously.

\begin{lemma}
\label{lem:gap2state}
For $k\geq 1$, $gap(\T^{k}_{AA}, \T^{k}_{BB}) \geq k+1$.
\end{lemma}
\begin{proof}
The case $k=1$ is proven by Claim \ref{clm:AA} and for higher $k$ we use
analogous arguments to the proof of Claim \ref{claim:smallgap} and Lemma
\ref{lem:gap}. We omit details.
\end{proof}

\begin{lemma}
\label{lem:gapdist2state}
For $k \geq 1$, $d^{2}_{MP}(\T^{k}_{AA}, \T^{k}_{BB}) = 12k$.
\end{lemma}
\begin{proof}
$\T^{k}_{AA}$ comprises $4k$ copies of $\T_a$. By taking
$4k$ copies of character $f_{asym}$, we see that
$d^{2}_{MP}( \T^{k}_{AA}, \T^{k}_{BB} ) \geq 4k(4-1) = 12k$. That $12k$
is also the upper bound, can be verified by showing $d_{rSPR}( \T^{k}_{AA}, \T^{k}_{BB}) \leq 12k$. This follows because by cutting off all copies of taxa $2, 5, 7$ into
separate components, we obtain an agreement forest of $\T^{k}_{AA}, \T^{k}_{BB}$
containing $12k + 1$ components.
\end{proof}


\subsection{The reduction}

We reduce from the NP-hard (and APX-hard) problem CUBIC MAX CUT \citep{alimonti2000}. Here we are given a cubic graph $G=(V,E)$, $|E| = 3|V|/2$,  and we are asked to partition $V$ into two disjoint pieces $V_1 \cup V_2$ such that
the number of edges that have one endpoint in $V_1$ and one endpoint in $V_2$ 
(``cut'' edges), is maximized. Let $MAXCUT(G)$ represent this value. We can assume without loss of generality that $G$ is connected and not bipartite.

The high-level idea is similar to the 2-state hardness reduction in \citep{fischer2014maximum}. Namely, we will construct two trees $\T_V$ and $\T_E$ and apply the symmetry-breaking gadget to ensure that for all optimal characters $f$, $l_f(\T_V) < l_f(\T_E)$. We will model the
vertices as subtrees in $\T_V$, each comprising three taxa, and argue - via a technical
argument - that these subtrees are monochromatic. We will let the 2 states represent the
two sides of the chosen partition $V_1 \cup V_2$. Henceforth we will call these states
\emph{red} and \emph{blue}. The colour of a vertex subtree thus denotes which side of
the partition it is on. The tree $\T_E$ will be constructed such that, the more cut edges
are induced by the partition chosen by $\T_V$, the higher the parsimony score of $\T_E$. The
construction will thus naturally choose a character that maximizes $MAXCUT(G)$. 

The fact that $\T_V$ and $\T_E$ must be binary, introduces significant complications
compared to the 2-state hardness reduction in \citep{fischer2014maximum}.
For this reason we will  introduce
two new special gadgets, that allow $\T_V$ (respectively, $\T_E$)
to be viewed as the independent union of several
subtrees. In $\T_V$ the gadget will be called the \emph{cherry switch} and in $\T_E$ we will have the $D(w_i)$ gadget, to be explained in due course. These independence gadgets neutralise the influence of side-effects that can occur as
a consequence of the fact that $\T_V$ and $\T_E$ are both binary.

\begin{figure}[ht]
\centering\vspace{0.5cm}
    \scalebox{.8}{ \includegraphics[width=14cm]{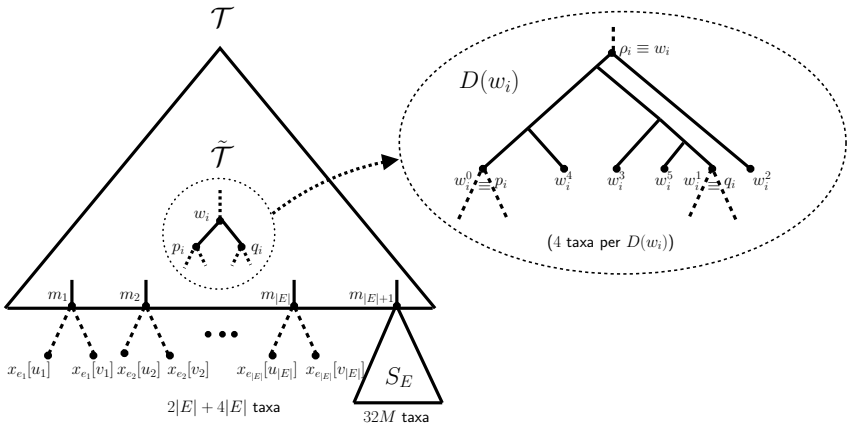}}
\caption{ Tree $\T$ is the left-hand side subtree of $\T_E$ in the 2-state NP-hardness construction, cf. Figure \ref{fig:binaryTE}. Every internal node $w_i$ with children nodes $p_i$ and $q_i$ of the original tree $\tilde{T}$ is replaced by tree $D(w_i)$ with root $w_i$ and children $w_i^0,\ldots,w_i^5$. Children $w_i^0$ and $w_i^1$ correspond to $p_i$ and $q_i$, respectively, whereas the other children form new leaves. Therefore, each $D(w_i)$ contributes four leaves to tree $\T$. For the leaves labelled $x_{e_i}[u_i]$ and $x_{e_i}[v_i]$, $u_i$ and $v_i$ are the endpoints
of edge $e_i$. }
\label{fig:lefthandside}
\end{figure}

We begin by constructing $\T_E$. First, we construct the left-hand side subtree $\T$ of $\T_E$ as depicted in Figure \ref{fig:lefthandside}. Let $\tilde{\T}$ be an arbitrary rooted binary tree
on $|E|+1$ leaves $\{m_1,\ldots,m_{|E|+1}\}$. Let $I = \{ w_1, \ldots, w_{|E|} \}$ be the $|E|$ interior nodes of $\tilde{\T}$. 
Let $M$ be a large integer whose value we will determine in due course. Let $S_V$
be the tree $\T^{M}_{AA}$ and $S_E$ be the tree $\T^{M}_{BB}$. Let $l$ be an arbitrary
leaf of $\tilde{\T}$. We replace $l$ with $S_E$. Next, select an arbitrary bijection between
the remaining leaves of $\tilde{\T}$ and $E$. For each edge $e = \{u,v\} \in E$, replace
the leaf of $\tilde{\T}$ corresponding to $e$ with a cherry on two taxa $\{x_e[u], x_e[v]\}$. Now,
for each internal vertex $w_i$, let $p_i$ and $q_i$ be the two children of $w_i$. We now introduce the independence gadget $D(w_i)$, constructed
as follows. Take a rooted binary tree $(w_i^{2},((w_i^{0},w_i^{4}),(w_{i}^{3},(w_i^{5},w_i^{1}))))$. We
replace $w_i$ with this tree, in the following sense: delete $w_i$, identify $w_i^{0}$ with $p_i$, identify $w_i^{1}$ with $q_i$
and if $w_i$ had an incoming edge, identify the root of $D(w_i)$ with the head of this edge. The
remaining leaves of $D(w_i)$ are $\{ w_i^{2}, w_i^{3}, w_i^{4}, w_i^{5} \}$ and we regard
these as taxa, so replacing each $w_i$ with $D(w_i)$ increases the number of taxa in total
by $4|E|$.
$\T$ has in total $2|E| + 32M + 4|E|$ taxa, where the $32M$ is the number of taxa in $S_E$.

\begin{figure} 
\center
\scalebox{.7}{\includegraphics{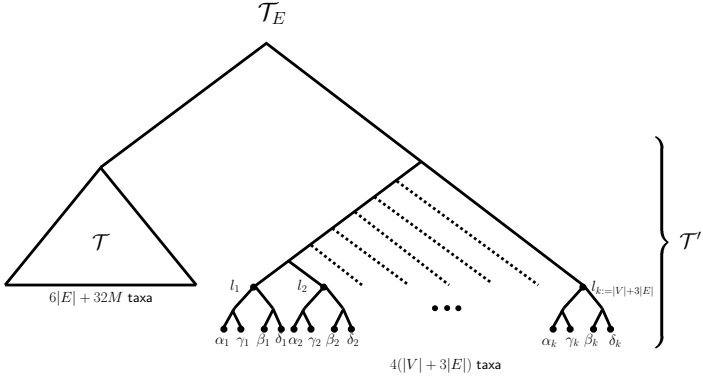} }
\caption{ Tree $\T_E$ for the 2-state NP-hardness construction consists of $\T$ as depicted in Figure \ref{fig:lefthandside} on the left-hand side and $\T'$ on the right-hand side. Note that $\T_E$ employs in total $32M+6|E|+4(|V|+3|E|)=32M+18|E|+4|V|$ taxa. }
\label{fig:binaryTE}
\end{figure}

Let $\T'$ be a rooted caterpillar on $|V| +  3|E|$ leaves $\{ l_1, \ldots, l_{|V|+3|E|}\}$. Replace
each leaf $l_i$ by a ``double cherry'' $( (\alpha_i, \gamma_i), (\beta_i, \delta_i))$ where
$\{ \alpha_i, \beta_i, \gamma_i, \delta_i \}$ are taxa. Join $\T$ and $\T'$ together by a new
root: this completes the construction of $\T_E$ as depicted in Figure \ref{fig:binaryTE}. $\T_E$ thus has in total,
\begin{align*}
&2|E| + 32M + 4|E| + 4( |V| + 3|E| )\\
&= 32M + 18|E| + 4|V| 
\end{align*}
taxa.

To construct $\T_V$ we start by creating a set of taxa-disjoint trees $J$. The disjoint union of the
taxa in the $|V| + 1 + 3|E|$ trees in $J$ will be exactly the set of taxa in the tree $T$ mentioned earlier. $J$ contains,
\begin{enumerate}
\item $S_V$;
\item For each vertex $u \in V$, a rooted triplet $( x_{e}[u], (x_{e^{*}}[u], x_{e^{**}}[u] ))$ where
$e, e^{*}, e^{**}$ are the three edges incident to $u$ in $G$;
\item For each gadget $D(w_i)$, two single taxon trees $w_{i}^{4}$ and $w_{i}^{5}$, and one
cherry $( w_i^{2}, w_i^{3} )$.
\end{enumerate}
Let $C$ be a rooted caterpillar on $|V| + 1 + 3|E|$ leaves. Consider a directed path on $|V| + 3|E|$ edges that
starts at the root of $C$ and terminates at one of the leaves in the unique cherry of $C$. Let $K$ be the
edges in this path. Choose an arbitrary bijection between
the leaves of $C$ and the trees in $J$, and replace each leaf with its corresponding subtree. We now
need to replace each edge in $K$ with a special gadget. In particular,
select an arbitrary bijection between $K$ and $\{ 1, \ldots, |V|+3|E| \}$. Next, for each edge in $K$,
subdivide it twice. From one of the vertices created by the subdivision operation, hang a cherry
$(\alpha_{i}, \beta_i)$, and from the other hang a cherry $(\gamma_i, \delta_i)$, where $i$ is the
index given by the bijection. We call these two cherries a \emph{cherry switch} - this is the independence gadget for $\T_V$. This completes the construction of $\T_V$, which is depicted in Figure \ref{fig:binaryTV}.

\begin{figure} 
\center
\scalebox{.7}{\includegraphics{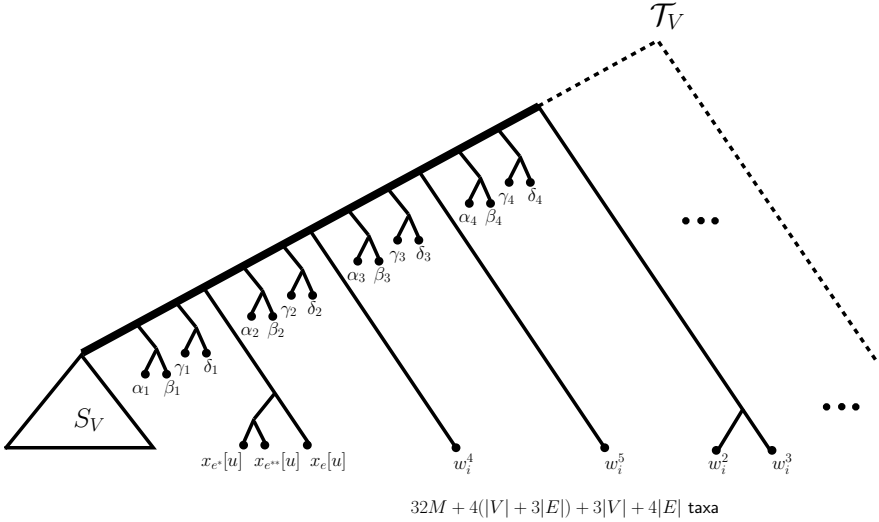} }
\caption{  $\T_V$ consists of a modification of a caterpillar tree with directed path $K$ which starts at the root and leads to a leaf in a cherry. Path $K$ is depicted in bold. Each of the $|V|+3|E|$ edges in $K$ carries two additional cherries $(\alpha_i,\beta_i)$ and $(\gamma_i,\delta_i)$. Therefore, $K$ contributes in total $4(|V|+3|E|)$ taxa. Moreover, for each of the original $|E|$ inner nodes $w_i$ of tree $\T$ as depicted in Figure \ref{fig:lefthandside}, $\T_V$ contains four taxa $w_i^2$, $w_i^3$, $w_i^4$, $w_i^5$. This leads to $4|E|$ more taxa. Finally, for each vertex $u$ in $V$, $\T_V$ contains a triple $(x_e[u],(x_{e^*}[u],x_{e^{**}}[u]))$, which are $3|V|$ taxa. Using the fact that in cubic graphs we have $|V|=\frac{2}{3}|E|$, $\T_V$ employs in total $32M+16|E| + 7|V| = 32M + 18|E| + 4|V|$ taxa. }
\label{fig:binaryTV}
\end{figure} 

We are now in a position to specify the number $M$. We require $M$ to be sufficiently
large that, for every optimal character $f$, $l_f( \T_V ) < l_f( \T_E )$. From
Lemma \ref{lem:gapdist2state} we know that there exists some character $f'$ such
that $l_f'( \T_E ) - l_f'(\T_V) \geq 12M$. We can obtain such a character
by extending the character suggested by Lemma \ref{lem:gapdist2state} such that all
taxa outside $S_E$ and $S_V$ are assigned the same state. Now, let $t$ be the number
of edges in $\T_V$ that lie outside $S_V$. For every character $f$ we have
\[
l_f( \T_V ) - l_f(\T_E) \leq t + (12M - (M+1)).
\]
The $12M$ term is obtained from Lemma \ref{lem:gapdist2state}, the $(M+1)$ term from Lemma \ref{lem:gap2state}, and the $t$ term arises (pessimistically) from the situation when every
edge in $\T_V$ (outside $S_V$) incurs a mutation, but no edge in $\T_E$ (outside $S_E$)
incurs a mutation. So, if we choose $M$ such that
\[
t + 12M - (M+1) < 12M
\]
it follows that for \emph{every} optimal character $f$, $d_{MP}(\T_V, \T_E) = l_f(\T_E) - l_f(\T_V)$ and in particular $l_f(\T_E) > l_f(\T_V)$. Choosing $M = t$ is therefore sufficient
to achieve this. 

Consider now the following.

\begin{observation}
\label{obs:diff}
Let $\T_1, \T_2$ be two binary trees and let $f$ be an optimal
character such that $l_f(\T_1) < l_f(\T_2)$.
 Suppose $\T_1$ contains two cherries $(a,b)$ and
$(c,d)$ and, in $\T_2$, there are cherries $(a,c)$ and $(b,d)$ under
a common parent (i.e. a ``double cherry''). Then $f$ can be modified to obtain an optimal character $f''$ 
in which $(a,b)$
and $(c,d)$ are both monochromatic but with different colours, and the colours
of all other taxa are unchanged.
\end{observation}
\begin{proof}
Let $f$ be an optimal character. We first apply Lemma \ref{lem:monochrome} to obtain an
optimal character $f'$ in which the two cherries are monochromatic in $\T_1$. If the
two cherries have different colours we are done. If not, then recolour one of the
cherries to obtain $f''$. This raises the parsimony score of $\T_1$ by (at most) one. In $\T_2$
two new mutations are created in the cherries $(a,c)$ and $(b,d)$ while
at most one mutation is saved on the edge entering the common parent. Hence,
$f''$ is also optimal.
\end{proof}

In exactly the same way as Theorem \ref{thm:hard} we now give an accumulating list of properties
which can be shown to be enjoyed by at least one optimal character that can be constructed in polynomial
time.

\noindent
\emph{\textbf{Property 1}. In $T_V$, for each $i \in |V| + 3|E|$, cherry
$(\alpha_i, \beta_i)$ is monochrome and cherry $(\gamma_i, \delta_i)$ is monochrome, and the
cherries have different colours.}\\
\\
\noindent
\emph{Proof.} This is an immediate consequence of Observation \ref{obs:diff}.\\
\\
\noindent
Next, observe that if a character $f$ has Property 1, and we swap the colours used
in some (or all) of the cherry switches to obtain $f'$, then $l_f(\T_E) = l_{f'}(\T_E)$. This is because
each cherry switch in $\T_V$ corresponds to a double cherry in $\T_E$, and (as long as Property 1 already
holds) the behaviour of the double cherries is invariant under permutation of \emph{red} and \emph{blue}.  This is the key
observation behind the next property. \\
\\
\noindent
\emph{\textbf{Property 2}. There is an optimal character $f$ such that}
\[
l_f(\T_V) = |V| + 3|E| + \sum_{T^{*} \in J} l_f( \T^{*} )
\]
\emph{where $l_f( \T^{*} )$ has the expected meaning i.e. the parsimony score of $\T^{*}$ after restricting $f$ to the taxa in $\T^{*}$.}\\
\\
\noindent
\emph{Proof.} Observe that for any optimal character $f$, $|V| + 3|E| + \sum_{\T^{*} \in J} l_f( \T^{*} )$ is
a lower bound on $l_f( \T_V )$. This can be observed by first applying Fitch's algorithm to the trees in $J$ (which
are all pendant in $\T_V$) and then noting that, due to Property 1, each of the $|V| + 3|E|$ cherry switches also incurs a mutation,
irrespective of the states that Fitch's algorithm designates to the roots of the trees in $J$. To show that it is also
an upper bound, first run Fitch on the trees in $J$. For those trees in $J$ that are allowed by Fitch to have either colour at the
root, pick one arbitrarily. For each cherry switch, consider the root state of the tree from $J$ directly above it (where here ``above'' means:
closer to the root of $\T_V$), and directly
below it. There are four possibilities: \emph{red-blue} (i.e. the tree from $J$ above it wants a root state of \emph{red}, the tree below it
wants \emph{blue}), \emph{red-red}, \emph{blue-red}, and \emph{blue-blue}.  If it is \emph{red-blue}, then if necessary swap the colours on the two cherries in the cherry switch, to ensure that the
\emph{red} cherry is closer to the root of $\T_V$. If it is \emph{blue-red}, then ensure that the \emph{blue} cherry is closer to the root of $\T_V$. Now, irrespective of which of the four possibilities holds, there is an optimal extension which occurs exactly one mutation (and not more) per cherry switch. In the \emph{red-blue} and \emph{blue-red} cases
the mutation will be on the edge between the two subdivision vertices (i.e. the edge between the vertices at which the two cherries
are attached to the caterpillar backbone). In the cases \emph{red-red} and \emph{blue-blue} the mutation will be
on the edge feeding into the \emph{blue}, respectively \emph{red} cherry.\\
\\
\noindent
\emph{\textbf{Property 3}. In $\T_V$, the trees in $J$ that are rooted triplets or cherries, are all monochrome.}\\
\\
\emph{Proof.} That the cherries can be made monochrome, is simply a consequence of Lemma \ref{lem:monochrome}. That the rooted triplets are monochrome
is more subtle. Consider any triplet in $J$, this has the form $( x_{e}[u], (x_{e^{*}}[u], x_{e^{**}}[u] ))$. We already know that $ \{ x_{e^{*}}[u], x_{e^{**}}[u] \}$ have the
same colour, as they form a cherry. Now, if $x_{e}[u]$ also has this colour, we are done. If not, then recolour it to give it the same colour as the other two
taxa. By Property 2, this \emph{must}
lower the parsimony score of $\T_V$ by exactly one. Hence, the new character is also optimal. (We really need Property 2 here, since ``the parsimony
score of $\T_V$ does not increase'' -- which in general is the strongest statement we can make after such a recolouring -- is not strong enough for our purposes).

Property 3 basically says that, in $\T_V$, the three taxa that represent each vertex of $G$ all have the same colour. This will allow us to encode MAX CUT
correctly. Property 3 also tells us that the $\{ w_i^{2}, w_i^{3} \}$
pairs of taxa, which form part of the $D(w_i)$ gadget, will be monochrome. This is particularly useful when combined with the fact that $w_i^{4}$ and $w_i^{5}$ are both
single taxa trees in $J$. A tree comprising only a single taxon has parsimony score 0, so whichever colour is allocated to the $w_i^{4}$ and $w_i^{5}$ taxa, they do
not impact upon the parsimony score of $\T_V$, by Property 2. In other words, these two taxa are ``free'': they can be allocated any colour in an attempt to cause as
many mutations as possible in $\T_E$. The $\{ w_i^{2}, w_i^{3} \}$ pairs of taxa are also ''free'', except for the limitation that $w_i^{2}$ and $w_i^{3}$ should have the same
colour. This underpins the following critical observation. 

\begin{observation}
\label{obs:dissolve}
Consider the rooted binary tree
\[
D(w_i) = (w_i^{2},((w_i^{0},w_i^{4}),(w_{i}^{3},(w_i^{5},w_i^{1})))).
\]
Suppose we fix $w_i^{0}$ as $red$, or $blue$, or $\{red, blue\}$, where $\{red,blue\}$ has the same meaning
as in Fitch's algorithm i.e. ``both states are possible''. Suppose we
do the same (independently) for $w_i^{1}$. Then depending on our choice
we can always select colours for $w_i^{2}$, $w_i^{3}$, $w_i^{4}$, $w_i^{5}$, whilst ensuring that the
same colour is chosen for $w_i^{2}$ and $w_i^{3}$, such that
the parsimony score of $D(w_i)$ under the resulting character is at least 2. Moreover, it is never possible
to achieve a parsimony score higher than 2 in this way.
\end{observation}
\begin{proof}
A straightforward case-analysis is sufficient to verify the ``at least 2'' part of the claim. There are $3^2$ cases, several
of which are symmetrical. These are the relevant cases:
\begin{enumerate}
\item $w_i^{0}$ and $w_i^{1}$ are both \emph{red}. Then choose all other taxa to be \emph{blue}.
\item  $w_i^{0}$ is \emph{red} and $w_i^{1}$ is \emph{blue}. Then choose $w_i^{4}$ to be \emph{blue},
$w_i^{5}$ to be \emph{red}, and $w_i^{2}$ and $w_i^{3}$ to both be \emph{blue}.
\item $w_i^{0}$ is \emph{red} and $w_i^{1}$ is $\{ red, blue \}$. Then choose $w_i^{4}$ to be \emph{blue},
$w_i^{5}$ to be \emph{blue}, and $w_i^{2}$ and $w_i^{3}$ to both be \emph{red}.
\item $w_i^{0}$ is $\{ red, blue \}$ and $w_i^{1}$ is \emph{red}. Then choose $w_i^{4}$ to be \emph{red},
$w_i^{5}$ to be \emph{blue}, and $w_i^{2}$ and $w_i^{3}$ to both be \emph{blue}.
\item $w_i^{0}$ and $w_i^{1}$ are both $\{ red, blue \}$. Then choose $w_i^{4}$ to be \emph{red}, $w_i^{5}$
to be \emph{red}, and $w_i^{2}$ and $w_i^{3}$ to both be \emph{blue}.
\end{enumerate}
To show that 3 or more mutations are never possible, note that a character on 6 taxa can only possibly have a parsimony
score of 3 if there are exactly 3 \emph{red} taxa and exactly 3 \emph{blue} taxa. (Otherwise, simply choose an
extension that assigns the majority colour to all internal nodes of the tree, yielding at most 2 mutations.) Now, if
at least one of $w_i^{0}$ and $w_i^{1}$ chooses $\{red, blue\}$, then 3 mutations are certainly not possible,
because we can (again) colour all the internal nodes of the tree monochrome in the majority colour, yielding
at most 2 mutations. So, suppose without loss of generality $w_i^{2}$ and $w_i^{3}$ are both \emph{red}. Then
exactly one of $w_i^{0}$ and $w_i^{4}$ will be \emph{red}, and the other \emph{blue}. But then
$w_i^{1}$ and $w_i^{5}$ will both be \emph{blue}. But this character has parsimony score at most 2, contradiction. 
\end{proof}

In $\T_E$ the taxa $w_i^{0}$ and $w_i^{1}$ become the roots of subtrees, and the three possible
choices for each taxon in Observation \ref{obs:dissolve} reflect the three possible decisions that
Fitch's algorithm can make when, in the bottom-up phase, the root of that subtree is reached. Essentially, then, Observation \ref{obs:dissolve}
allows us to ``glue'' these two subtrees together with a profit of \emph{exactly} 2 mutations, entirely independently
of the two subtrees themselves. 

Now, consider any optimal character $f$ that has Property 3 (and thus all earlier properties too). We have
\[
l_f(\T_V) = |V| + 3|E| + l_f( S_V )
\]
since (by Property 2) the singletons, cherries and triplets in $J$ do not internally generate any
mutations and mutations along the $K$ part of $\T_V$ are already accounted for. (As usual, $ l_f( S_V )$ refers to the parsimony score of the restriction of $f$ to the taxa in $S_V$). Let $CUT(f)$ be the number of cut edges induced by $f$ i.e. after partitioning the vertices of $V$ according to the colours of the corresponding rooted triplets in $J$. We have,
\[
l_f(\T_E) = 2(|V| +  3|E|) + l_f(S_E) + CUT(f) + 2|E| 
\]
The $2(|V| +  3|E|)$ term is the contribution of the double cherries, and the $2|E|$ term is the 2 mutations that we know we can definitely incur in each
$D(w_i)$ gadget. Hence, an optimal character should try and make the induced cut as large as possible: there is no other freedom. Consequently,
\begin{align*}
d^{2}_{MP}(\T_V, \T_E) & = l_f(\T_E) - l_f(\T_V)\\
&=  2(|V| +  3|E|) + l_f(S_E) + MAXCUT(G) + 2|E| - (|V| + 3|E| + l_f( S_V ))\\
&= |V| + 5|E| + (l_f(S_E) - l_f( S_V)) + MAXCUT(G)\\
&= |V| + 5|E| + 12M + MAXCUT(G)
\end{align*}
The fact that $(l_f(S_E) - l_f( S_V))$ is equal to $12M$ is not entirely
automatic. It is a consequence of the fact that in this context there is no point choosing a character $f$ which, when restricted to $S_V$ and $S_E$, yields an MP distance smaller
than $d^{2}_{MP}(S_V, S_E)$ (where the latter value is equal to 12M by Lemma \ref{lem:gapdist2state}).

The terms can easily be rearranged to obtain $MAXCUT(G)$ from $d^{2}_{MP}$, which yields the overall theorem:

\begin{theorem}
Computation of $d^{2}_{MP}$ is NP-hard on binary trees.
\end{theorem}

We also obtain the following corollary.

\begin{corollary}
Computation of $d^{2}_{MP}$ is APX-hard on binary trees.
\end{corollary}
\begin{proof}
We will show that  if $d^{2}_{MP}$ can be approximated
in polynomial time to within a multiplicative factor of $(1-\epsilon)$, for some $\epsilon>0$, that CUBIC MAXCUT can be approximated in polynomial time to within a factor of $(1-k\epsilon)$ for some constant $k>0$ that is independent of $\epsilon$. Given that CUBIC MAX CUT is APX-hard \citep{alimonti2000} there is (by definition) some $\epsilon' > 0$ such that a factor $(1-\epsilon')$ approximation or better is not possible in polynomial time unless $P=NP$.  The APX-hardness of
$d^2_{MP}$ will then follow\footnote{Formally speaking we should give an L-reduction here \citep{papadimitriou1991optimization}. For brevity we omit
the technicalities. An L-reduction can if desired easily be constructed from the information provided here.} : the corresponding threshhold for $d^{2}_{MP}$ will be $\epsilon' / k$.

First, suppose we obtain character $f$, which is a $(1-\epsilon)$ approximation to $d^2_{MP}( \T_V, \T_E)$. We need to show that
a feasible solution (i.e. a cut) can be extracted in polynomial time from $f$, which requires that the solution obeys all the Properties. Character $f$ might not have these Properties, but they can be acquired in polynomial time without lowering the parsimony distance score of the character. To do this, ensure first that $S_V$ and $S_E$ use the
duplicated character $f_{asym}$ (which optimizes the MP distance between $S_V$ and $S_E$). This ensures that $l_f(\T_V) < l_f(\T_E)$. From this point on the Properties can
be accumulated one at a time: the constructive proofs describing how the Properties are obtained do not require that $f$ is
optimal, only that $l_f(\T_V) < l_f(\T_E)$.

Recall that $|V| =(2/3)|E|$. We need an explicit expression for $M$. This was set to be $t$, the number of edges
in $\T_V$ minus the edges in subtree $S_V$. $\T_V$ has in total $32M + 18|E| + 4|V|$ taxa, and after subtracting
the $32M$ this gives $18|E| + 4|V|$. A rooted binary tree on $|X|$ taxa has $2(|X|-1)$ edges, yielding
$36|E| + 8|V| - 2$, plus 2 extra edges created when the subtree $S_V$ is re-attached, giving $36|E| + 8|V|$
which is $(124/3)|E|$. Hence,
\begin{align*}
d^{2}_{MP}(\T_V, \T_E) &=  |V| + 5|E| + 12M + MAXCUT(G)\\
&= (2/3)|E| + 5|E| + 496|E| + MAXCUT(G)\\
&= (1505/3)|E| + MAXCUT(G).
\end{align*}
The size of the cut returned after processing $f$ is at least
\begin{align*}
&= ( 1-\epsilon )( (1505/3)|E| + MAXCUT(G) ) - (1505/3)|E|\\
& = (1-\epsilon)MAXCUT(G) - \epsilon(1505/3)|E|
\end{align*}

It is well-known
that for cubic $G$, $MAXCUT(G) \geq 2|E|/3$, by moving a vertex to the other side of the partition if one or fewer
of its incident edges is in the cut. So,
\begin{align*}
& (1-\epsilon)MAXCUT(G) - \epsilon(1505/3)|E|\\
& \geq (1-\epsilon)MAXCUT(G)  - \epsilon(1505/2)MAXCUT(G)\\
& = (1 - \frac{1507}{2}\epsilon)MAXCUT(G).\\
\end{align*}
This concludes the proof.
\end{proof}


\section{An Integer Linear Programming (ILP) formulation for binary instances}
\label{sec:ilp}
Let $\T_1$ and $\T_2$ be two binary phylogenetic trees on $n \geq 2$ taxa.
Given the hardness of MP distance it is natural to ask how well $d_{MP}(\T_1, \T_2)$ can be computed in practice.  One option is to leverage the result in \citep{fischer2014maximum} which proves that there always exists an optimal character that is convex on one of the trees (i.e. has a parsimony score exactly one less than the number of states in the character). Hence we can guess which of the two input trees is convex, guess the number of states $s$ in the optimal character, and then guess the $(s-1)$ edges of the convex tree on which the mutations occur.  Assuming the trees are unrooted, and letting $g(\T_1, \T_2)$ be any safe upper bound on $s$, this gives a deterministic running time of
\[
O\bigg( \sum_{s=2}^{g(\T_1,\T_2)} \binom{2n-3}{s-1} \bigg ).
\]
As the following observation shows, we can take $g(\T_1, \T_2) = \lfloor n/2 \rfloor$.
\begin{observation}
Let $\T_1$ and $\T_2$ be two binary phylogenetic trees on $n \geq 2$ taxa. There exists
an optimal convex character with at most $\lfloor n/2 \rfloor$ states. Moreover, this bound
is tight.
\end{observation}
\begin{proof}
Let $f$ be an optimal convex character. Suppose $f$ has strictly more than $\lfloor n/2 \rfloor$ states. Then there exists a state $t$ that occurs on only one taxon $x$. We root $\T_1$ on the edge entering $x$. If we run Fitch on this rooted tree a union event will necessarily be generated at the root due to the fact that $t$ occurs on only one taxon. Let $C$ be the set of states in this union event, and let $t'$ be any state in $C \setminus \{t\}$. Let $f'$ be the character obtained
from $f$ by assigning state $t'$ to taxon $x$. By re-running Fitch we see that $l_{f'}(\T_1) = l_{f}(T_1)-1$. Moreover, $f'$ has one fewer state than $f$, so $f'$ is convex. By
Observation \ref{obs:mostone} $l_{f'}(\T_2) \geq l_{f}(\T_2) - 1$. Hence, $f'$ is
optimal, convex and has fewer states than $f$. By repeating this process we eventually obtain
an optimal convex character with at most $\lfloor n/2 \rfloor$ states.

The trees in Figure \ref{fig:binaryAntisymunbounded} on 6 taxa are a tight example for this bound: it can easily be verified computationally that for these two trees optimal characters require at least 3 states.
\end{proof}


Of course, even if we take $g(\T_1, \T_2) = \lfloor n/2 \rfloor$, such brute-force algorithms will quickly become impractical for even very small $n$. Hence we turn to
Integer Linear Programming (ILP), which allows us to compute $d_{MP}$ and $d^{i}_{MP}$ for 
larger trees. The ILP for computing $d^{2}_{MP}$ performs very well, allowing computation of $d^{2}_{MP}$ in reasonable time for trees with up to 100 taxa. Unfortunately, in the case of  $d_{MP}$ the ILP struggles to terminate in reasonable time for trees with more than 16 taxa. Future research (i.e. better ILP formulations) will hopefully improve upon this.

The ILP formulation is currently limited to binary trees but the model could be extended to non-binary
trees without too much difficulty.

Let $\T_1$ and $\T_2$ be rooted, binary phylogenetic trees on the same set of taxa $X$,
where $|X|=n$. Let $U$ be the internal nodes of $\T_1$ and $V$ the internal nodes of $\T_2$. Let $s$ be a constant denoting the maximum number of states that any character can have; as discussed taking $s = \lfloor n/2 \rfloor$ is a safe choice.
(To compute
$d^{i}_{MP}$ we simply take $s \leq i$.)
The
following ILP maximizes $l_f(\T_1) - l_f(\T_2)$ ranging over all characters $f$ with at most
$s$ states. To obtain the true parsimony distance the ILP should be run twice, once to
compute the maximum of $l_f(\T_1) - l_f(\T_2)$ and once to compute the maximum of $l_f(\T_2) - l_f(\T_1)$.\\
\\
All variables in the program are binary.\\
\\
First of all we constrain that in both trees the taxa have the same state, and that each taxon
chooses exactly one state. We introduce variables $x_{t,i}$ for each $t \in X$ and
$1 \leq i \leq s$. For each $t \in X$ we introduce the constraint:
\[
\sum_{i=1}^{s} x_{t,i} = 1
\]
We now show how the parsimony score can be computed for $T_1$. The variables and constraints essentially ``hard-code'' Fitch's algorithm. (The encoding of $T_2$ is symmetrical. The
two encodings are linked together via the variables that represent the states of the taxa in $X$, and the objective function, which we shall discuss in due course).\\
\\
Given an internal node $u \in U$, let $l$ be its left child and $r$ be its right child. Fitch's
algorithm tells us to take the intersection of the states at $l$ and $r$, if the intersection
is non-empty, and otherwise the union (in which case we pay 1 mutation). We do this
computation as follows. For $1 \leq i \leq s$ we introduce a variable $x_{u,i}$. The idea
is that $x_{u,i}$ will be 1 if and only if state $i$ is in the set of states at node $u$ (in the
bottom-up phase of Fitch). We determine the set of states at $u$ by performing the
union and intersection computations directly. For that purpose, for $1 \leq i \leq s$ we introduce $x_{u,i}^{\cap}$ and $x_{u,i}^{\cup}$ and the following constraints:
\begin{align*}
x_{u,i}^{\cap} &\leq x_{l,i}\\
x_{u,i}^{\cap} &\leq x_{r,i}\\
x_{u,i}^{\cap} &\geq x_{l,i} + x_{r,i} - 1\\
\\
x_{u,i}^{\cup} &\geq x_{l,i}\\
x_{u,i}^{\cup} &\geq x_{r,i}\\
x_{u,i}^{\cup} &\leq x_{l,i} + x_{r,i}
\end{align*}
The top group of constraints ensure that the $x_{u,i}^{\cap}$ variables reflect the
intersection of the states at the children (i.e. logical AND) and $x_{u,i}^{\cup}$ the union (i.e. logical OR).\\
\\
For each $u \in U$ we have a variable $x_{u}^{\cap}$ and $x_{u}^{\cup}$ which is 1 (resp. 0) if Fitch wants an intersection
operation at node $u$. We can ensure that these variables take
the correct value as follows. Firstly:
\[
\sum_{i=1}^{s} x_{u,i}^{\cap} \geq x_{u}^{\cap}
\]
And, secondly, we add the following constraint for each $1 \leq i \leq s$:
\[
x_{u}^{\cap} \geq x_{u,i}^{\cap}
\]
To ensure that $x_{u}^{\cap}$ and $x_{u}^{\cup}$ are complementary we add the constraint
\[
x_{u}^{\cap} + x_{u}^{\cup} = 1
\]

Now, we have to ensure that $x_{u,i}$ takes the value $x_{u,i}^{\cap}$ whenever $x_{u}^{\cap}$ is 1, and $x_{u,i}^{\cup}$ otherwise. We do this by, for each $1 \leq i \leq s$,
adding the following four constraints:
\begin{align*}
x_{u,i} &\geq x_{u,i}^{\cap}\\
x_{u,i} &\leq x_{u,i}^{\cup}\\
x_{u,i} &\leq x_{u,i}^{\cap} + 1 - x_{u}^{\cap}\\
x_{u,i} &\geq x_{u,i}^{\cup} - x_{u}^{\cap}
\end{align*}
Finally, all that remains is to compute the difference between the two parsimony scores. We do this with the following objective function:
\[
\text{Maximize } \sum_{u \in U} x_{u}^{\cup} - \sum_{v \in V} x_{v}^{\cup}
\]
This concludes the formulation. We have implemented it by using Java to translate the
input trees into an ILP format suitable for solvers such as GLPK, SCIP or CPLEX. We have
used this to verify several of the bounds used in Section \ref{subsec:symbreak}. The source code can be downloaded from \citep{MPdistILP}.

We tested our ILP running CPLEX on a 3.10GHz 64-bit machine with 4 GB RAM.
We observed the following running times. For computation of $d_{MP}$:
\begin{itemize}
\item the two trees $\T_a$ and $\T_b$ on 6 taxa as depicted in Figure \ref{fig:binaryAntisymunbounded}: total running time $<$ 1 second.
\item the two trees $\T_A$ and $\T_B$ on 12 taxa consisting of two copies of $\T_a$ or $\T_b$, respectively: 70 seconds. 
\end{itemize}
For computation of $d^{2}_{MP}$: 
\begin{itemize}
\item the two trees on 8 taxa as depicted in Figure \ref{fig:binaryAntisym2states}:  $<$ 1 second.
\item the two trees $\T_A$ and $\T_B$ on 16 taxa consisting of two copies of $\T_a$ or $\T_b$, respectively: $<$ 1 second.
\item the two trees $\T_{AA}$ and $\T_{BB}$  on 32 taxa consisting of four copies of $\T_a$ or $\T_b$:  6 seconds.
\end{itemize}
\par  

Computation of $d^{i}_{MP}$, for small $i$, is much faster than $d_{MP}$ due to the greatly reduced number of binary variables. We observed that the ILP could compute $d^{2}_{MP}$ for trees with 100 taxa in approximately 140 seconds.

\section{Conclusion}
\label{sec:conclusion}

In this article we have proven that calculating MP distance ($d_{MP}$) is NP-hard on binary trees. Computation of $d^{2}_{MP}$ (the version of the problem where we are restricted to binary characters) is also NP-hard on binary trees. The latter problem is also APX-hard, and determining whether $d_{MP}$ is APX-hard remains an open question. At the moment we do not have an NP-hardness proof for $d^{3}_{MP}$ on binary trees but given that $d^{i}_{MP}$ on binary trees is NP-hard for each $i \geq 4$ we expect that this will also be hard.

We have presented and implemented a simple ILP formulation, which is publicly available at \citep{MPdistILP}. The ILP is much faster than obvious brute-force algorithms and allowed us to verify the MP-distance of the symmetry-breaking gadgets used in the hardness reductions.  The ILP for $d^{2}_{MP}$ is fast but the ILP for $d_{MP}$ does not scale well. An important open problem is therefore to develop an ILP formulation that avoids the present approach of simply hard-coding Fitch's algorithm.

Finally, elucidating the exact relationship between MP distance and other phylogenetic metrics
remains an intriguing challenge.

\section*{Acknowledgement}
We would like to thank Nela Lekic for valuable discussions.

\section{Appendix }\label{sec:appendix}
Here, we present the proofs we omitted in the previous sections. 

\emph{\textbf{Property 2.1.} In $\T_E$, the cherry $\{ \beta_1, \beta_2 \}$ has a different
colour to the cherry $\{ \gamma_1, \gamma_2 \}$.}\\
\\
\noindent
\emph{Proof.} Suppose this is not so. Recolour
$\{ \beta_1, \beta_2 \}$ to some new colour not appearing elsewhere. This increases the
number of mutations in $\T_E$ by at most 1. However, in $\T_V$ the number of mutations
in the $\beta_1, \beta_2, \gamma_1, \gamma_2$ subtree increases from 0 to 2. Possibly
$\T_V$ then saves a single mutation at the root, but in any case the parsimony score of $\T_V$
increases by at least 1. So the new character is still optimal.\\
\\
\noindent
\emph{\textbf{Property 2.2.} In $\T_E$, the (possibly multiple) colours used for the taxa of $B$ (including $\alpha$) are not used
elsewhere in $\T_E$, except possibly $\{ \beta_1, \beta_2 \}$. }\\
\\
\noindent
\emph{Proof.}
Take an optimal extension $F$ of $f$ by applying Fitch's algorithm. Let $c$ be the colour allocated to the root of $B$ by this
extension.  Let $c^{*}$ be the colour of the parent $p_1$ of
the root of $B$, and $c^{**}$ the colour of its parent $p_2$.  Let $c_{\beta}$ be the
colour of the $\{ \beta_1, \beta_2 \}$ taxa and define $c_{\gamma}$ similarly.
Suppose $c = c^{*} = c^{**}$.  We will recolour the character -- and this extension --
to ensure that this is no longer the case. By Property 2.1, $c_{\beta} \neq c_{\gamma}$. If
$c \neq c_{\beta}$ and $c \neq c_{\gamma}$, then recolour $\gamma_1, \gamma_2$ and
their parent to colour $c$.  (This lowers the parsimony score of $\T_E$ by 1, and can lower the parsimony score of $\T_V$ by at most 1, so the character - and the extension - is still optimal.) Otherwise, exactly one of $c_{\beta}$ and $c_{\gamma}$ is equal to $c$. If $c_{\beta}$ has this property, then swap the colours of $\{ \beta_1, \beta_2 \}$ and $\{ \gamma_1, \gamma_2 \}$ (and their parents). So we now have $c = c_{\gamma}$ and $c \neq c_{\beta}$. In particular, there is a mutation on the edge entering
the cherry $\{\beta_1, \beta_2\}$. For technical reasons we now introduce a brand new colour, \emph{bronze} say,
and recolour $\{\beta_1, \beta_2\}$ (and their parent) to be \emph{bronze}. This leaves the parsimony score of $\T_E$ unchanged, and cannot decrease the parsimony score of $\T_V$, so
the character is still optimal. We do this simply to ensure that the colour of $\beta_1, \beta_2$ does not occur anywhere else. Run Fitch's algorithm on $\T_V$ and record the output
as $R$. 

At this point we introduce a new colour \emph{silver}. 
Recolour the following vertices silver: $\beta_1, \beta_2$, their parent, $p_1$ and the entire $c$-coloured connected component inside $B$ starting at the root of $B$. This gives a new character and extension which saves one mutation (on the edge leading into the cherry $\beta_1, \beta_2$) but creates one mutation between $p_1$ and $p_2$. So the
parsimony score of $\T_E$ does not increase. It is
not obvious, but the parsimony score of $\T_V$ will not drop. To see why this is,
note that (under this particular recolouring) the only way the parsimony score of $\T_V$ could drop, is if the recolouring causes a mutation (i.e. union event) at the root of $\T_V$ to vanish, and at the same time does not create any additional mutations elsewhere. If $R$ did not have a mutation event at the root of $\T_V$ anyway we are done, there is nothing to consider. If it did, then in $R$ the union event at the root must have had the form $\{ c, bronze \} \cup W$ where
$W \cap \{c, bronze\} = \emptyset$ and the $W$ is the set of states generated by the bottom-up phase of Fitch's algorithm for the root-incident right subtree of $\T_V$, let us call this $T_{right}$. Now, if the recolouring causes the parsimony
score of $T_{right}$ to increase, we are also done. So suppose the parsimony score
of $T_{right}$ stays the same and $T_{right}$ suddenly has an optimal extension (generated by any method, not necessarily Fitch) in which its
root can be coloured $c$ or $silver$ (which is necessary to save a mutation at the root
of $\T_V$). But then we could take this extension and re-merge the colours $c$ and \emph{silver} back into $c$, showing that $T_{right}$ did originally have an optimal extension in which its root could be coloured $c$. This would mean that $R$ cannot possibly have been
an optimal extension: it claimed a mutation was needed at the root of $\T_V$, but we have just shown that colouring the root $c$ would have avoided mutations on both of its outgoing edges. Contradiction on the assumed optimality of $F$.

Hence, this new character is indeed still optimal. The modified extension
(on $\T_E$) is necessarily also optimal for this new character: if some other extension existed that
induced fewer mutations, then this would violate the assumed optimality of the original
character (i.e. because the parsimony score of $\T_V$ does not decrease).

At this point we can recolour all the monochromatic connected components induced by the extension, and starting at some
vertex of $B$, with brand new colours. This new character must be optimal. (The
score of $\T_V$ under this new character does not decrease, so the recoloured extension must
also be optimal.) Moreover, with the possible exception of $\beta_1, \beta_2$ none of the
colours used for taxa in $B$ are used outside $B$. This is guaranteed because the \emph{silver}
recolouring ensured that there are no longer monochromatic connected components that
connect taxa in $B$ with taxa beyond $B \cup \{ \beta_1, \beta_2 \}$.\\
\\
\noindent
\emph{\textbf{Property 2.3} In $\T_E$, all the taxa in $B$ have the same colour which, with
the possible exception of $\beta_1,\beta_2$, does not appear on taxa outside $B$ and $\alpha$.}\\
\\
\emph{Proof.} Let $f$ be an optimal character. If the taxa in $B$ are monochromatic we are done. Otherwise, run Fitch's algorithm to generate an optimal extension on $\T_E$. (Also run Fitch on $\T_V$ and let $m$ be the number of mutations incurred there, although we do not need to remember the corresponding extension). In $\T_E$ at least one node of $B$ must be a union event (in the bottom-up phase of Fitch). Let $u$ be such a node that is furthest from the root
of $B$, and let $\T_u$ be the subtree of $B$ rooted at $u$. Let $\T_1, \T_2$ be the two subtrees rooted at the two children of $u$. The taxa in
$\T_1$ must be monochromatic with some colour $c_1$, and the taxa in $\T_2$ must be monochromatic with some colour $c_2 \neq c_1$. Suppose, without loss of generality, that
the optimal extension colours $u$ with colour $c_2$. This causes a mutation between
$u$ and the root of $\T_1$. Hence, if we recolour the entire subtree $\T_1$ (i.e. taxa and
non-taxa alike) with colour $c_2$, then this generates a new character $f'$ (and new extension) in which the parsimony score of $\T_E$ drops by (at least) 1. We argue that $f'$
can decrease the parsimony score of $\T_V$ by (at most) 1, from which the optimality
of $f'$ (and its new extension) will follow. Suppose, for the sake of contradiction, that
$f'$ generates $m-2$ or fewer mutations in $\T_V$. Apply Fitch to $f'$ on
$\T_V$. Now, due to the fact that $B$ has essentially the same topology in both $\T_V$ and
$\T_E$, the subtree $\T_u$ is topologically preserved inside $\T_V$. In particular, the
images of all vertices of $\T_u$ are unambiguously defined inside $\T_V$. Now, in its bottom-up phase Fitch will generate in $\T_V$ no union events on the images of the nodes of $\T_u$, due to the fact that all taxa in $\T_u$ have colour $c_2$. (There might
be a union event generated at the point a pendant rooted triplet is grafted onto the image of $\T_u$, see Figure \ref{fig:TVunbounded}, but such subdivision nodes are not considered to be part of the image of $\T_u$.)
At this point we recolour in $\T_V$
$\T_1$ (taxa and non-taxa alike) with colour $c_1$, creating in total exactly one extra mutation, on the edge between $u$ and the root of $\T_1$.\footnote{It is possible that for some taxa $x$ in $\T_1$  the recolouring of $\T_1$ back to $c_1$ causes a mutation in $\T_V$ to \emph{move} from the parent edge of $x$ to its sibling edge, but the mutation will not disappear, due to the fact that by Property 2.2 $x$ has a different colour to all taxa in its sibling subtree, and the fact that we generated the original extension using Fitch's algorithm.}  This new extension is a valid extension of $f$ on $\T_V$ but generates at most $m-1$ mutations, contradicting the assumption that
an optimal extension of $f$ on $\T_V$ had $m$ mutations. Hence, $f'$ must be optimal. 

If $f'$ is not yet monochromatic for $B$, then we re-run Fitch on $\T_E$ to generate a fresh optimal extension, and iterate the entire process until $B$ becomes monochromatic. This process
must terminate (in polynomial time) because each iteration merges two distinctly coloured
subtrees of $B$ into one strictly larger monochromatic subtree.

\noindent
\emph{\textbf{Property 3.} In $\T_E$, all the taxa in $B$ (including $\alpha$) have the same colour, and cherry
$\{\beta_1, \beta_2\}$ also has this colour. Moreover this colour does not appear on any
other taxa i.e. it is unique for $B$ and $\beta_1, \beta_2$.}\\
\\
\noindent
\emph{Proof.} From Property 2.3 we already know that all taxa in $B$ have the same colour and, with
the possible exception of $\beta_1, \beta_2$, this colour does not appear outside $B$. Let
$c$ be the colour used in $B$. If $c$ is the same as the colour of $\{ \beta_1, \beta_2 \}$,
denoted again $c_{\beta}$, we are done. If $c$ is the same colour as $\{ \gamma_1,
\gamma_2 \}$, then swapping the colours on $\{ \beta_1, \beta_2 \}$ and
$\{ \gamma_1, \gamma_2 \}$ preserves optimality, and we are done. (Optimality
is preserved because the parsimony score
of $\T_E$ cannot increase under such a swap, and the parsimony score of $\T_V$ cannot decrease due to symmetry.) So suppose neither $c_\beta$ nor $c_\gamma$ is equal to $c$. Run
Fitch to generate an optimal extension. In the bottom-up phase Fitch will assign states $\{c, c_{\beta}\}$
to $p_1$ and $\{c, c_{\beta}, c_{\gamma}\}$ to $p_2$. Suppose, in the top-down phase, the parent of $p_2$ communicates a state to $p_2$ that is either equal to $c$, or not in $\{c, c_{\beta}, c_{\gamma}\}$. In this case
Fitch allows us to give $p_2$ colour $c$. We can then recolour $\{ \gamma_1, \gamma_2 \}$ to be $c$ (saving
at least one mutation in $\T_E$, and saving at most one mutation in $\T_V$, thus
preserving optimality) and then switch back to the earlier case. If Fitch permits $p_1$ to
be coloured $c$, we simply recolour $\{ \beta_1, \beta_2 \}$ to be $c$ and we are done
because this, via the same analysis, preserves optimality. 
 The only case remaining is if every possible set of choices in the top-down phase of Fitch leads to the
conclusion that both $p_1$ and $p_2$ are coloured $c_\beta$. (This is the only remaining case
because if $p_2$ is or can be coloured $c_{\gamma}$, then Fitch will subsequently allow us to colour $p_1$ with colour $c$,
due to the fact that $c_{\gamma} \not \in \{c, c_{\beta}\}$ i.e. we will be in an earlier case.) So consider an
extension generated by Fitch in this case. We swap the colours
on $\{ \beta_1, \beta_2 \}$ and $\{ \gamma_1, \gamma_2 \}$ (including the colours
of their parents). This colour swap does not affect the number of mutations but it ensures that both edges leaving
$p_1$ carry mutations. Hence, if we now colour $\{ \beta_1, \beta_2 \}$, their parent, and $p_1$
all $c$, both these mutations vanish. So we definitely save one mutation in $\T_E$, and as usual at most one mutation
is saved in $\T_V$. So we are done. This concludes the proof of Property 3.

\bibliographystyle{plainnat}
\bibliography{bibliographyMPdistanceBinary}

\begin{thebibliography}{14}
\providecommand{\natexlab}[1]{#1}
\providecommand{\url}[1]{\texttt{#1}}
\expandafter\ifx\csname urlstyle\endcsname\relax
  \providecommand{\doi}[1]{doi: #1}\else
  \providecommand{\doi}{doi: \begingroup \urlstyle{rm}\Url}\fi

\bibitem[Alimonti and Kann(2000)]{alimonti2000}
P.~Alimonti and V.~Kann.
\newblock Some {APX}-completeness results for cubic graphs.
\newblock \emph{Theoretical Computer Science}, 237\penalty0 (1--2):\penalty0
  123 -- 134, 2000.

\bibitem[Bordewich and Semple(2005)]{bordewichsemple2005}
M.~Bordewich and C.~Semple.
\newblock On the computational complexity of the rooted subtree prune and
  regraft distance.
\newblock \emph{Annals of Combinatorics}, 8\penalty0 (4):\penalty0 409--423,
  2005.

\bibitem[Diestel(2005)]{diestel2005}
R.~Diestel.
\newblock \emph{Graph theory (graduate texts in mathematics)}.
\newblock Springer, 2005.

\bibitem[Felsenstein et~al.(2000)Felsenstein, Archie, Day, Maddison, Meacham,
  Rohlf, and Swofford]{felsenstein_2000}
J.~Felsenstein, J.~Archie, W.~Day, W.~Maddison, C.~Meacham, F.~Rohlf, and
  D.~Swofford.
\newblock The newick tree format., 2000.
\newblock URL
  \url{http://evolution.genetics.washington.edu/phylip/newicktree.html}.

\bibitem[Fischer and Kelk(2014)]{fischer2014maximum}
M.~Fischer and S.~Kelk.
\newblock On the maximum parsimony distance between phylogenetic trees.
\newblock \emph{Annals of Combinatorics}, 2014.
\newblock preliminary version arXiv preprint arXiv:1402.1553.

\bibitem[Fischer and Thatte(2009)]{fischerthatte2009}
M.~Fischer and B.~Thatte.
\newblock Revisiting an equivalence between maximum parsimony and maximum
  likelihood methods in phylogenetics.
\newblock \emph{Submitted to Bulletin of Mathematical Biology}, 2009.

\bibitem[Fitch(1971)]{fitch_1971}
W.~Fitch.
\newblock Toward defining the course of evolution: minimum change for a
  specific tree topology.
\newblock \emph{Systematic Zoology}, 20\penalty0 (4):\penalty0 406--416, 1971.

\bibitem[Haws et~al.(2013)Haws, Hodge, and Yoshida]{haws2013}
D.~Haws, T.~Hodge, and R.~Yoshida.
\newblock Phylogenetic tree reconstruction: geometric approaches.
\newblock In R.~Robeva and T.~Hodge, editors, \emph{Mathematical concepts and
  methods in modern biology -- using modern discrete models.}, pages 307--342.
  Elsevier, 2013.

\bibitem[Holyer(1981)]{holyer1981np}
I.~Holyer.
\newblock The {N}{P}-completeness of edge-coloring.
\newblock \emph{SIAM Journal on Computing}, 10\penalty0 (4):\penalty0 718--720,
  1981.

\bibitem[Huson and Steel(2004)]{husonsteel}
D.~Huson and M.~Steel.
\newblock Distances that perfectly mislead.
\newblock \emph{Systematic Biology}, 53(2):\penalty0 327 -- 332, 2004.

\bibitem[Iersel et~al.(2013)Iersel, Kelk, Leki{\'c}, and
  Scornavacca]{lv2013practical}
L.v. Iersel, S.~Kelk, N.~Leki{\'c}, and C.~Scornavacca.
\newblock A practical approximation algorithm for solving massive instances of
  hybridization number for binary and nonbinary trees.
\newblock \emph{BMC Bioinformatics}, 15:\penalty0 127--127, 2013.

\bibitem[Kelk and Fischer(2014)]{MPdistILP}
S.~Kelk and M.~Fischer.
\newblock Maximum parsimony distance integer linear program {(MPDIST)}.
\newblock \url{http://skelk.sdf-eu.org/mpdistbinary/}, 2014.

\bibitem[Maddison(1989)]{maddison1989reconstructing}
W.~Maddison.
\newblock Reconstructing character evolution on polytomous cladograms.
\newblock \emph{Cladistics}, 5\penalty0 (4):\penalty0 365--377, 1989.

\bibitem[Papadimitriou and Yannakakis(1991)]{papadimitriou1991optimization}
C.H. Papadimitriou and M.~Yannakakis.
\newblock Optimization, approximation, and complexity classes.
\newblock \emph{Journal of computer and system sciences}, 43\penalty0
  (3):\penalty0 425--440, 1991.

\end{thebibliography}

\end{document}